\documentclass[aps,prl,preprint,notitlepage,groupedaddress,floatfix]{revtex4-2}

\usepackage{amsmath,amssymb,bm}
\usepackage{graphicx}
\usepackage{dcolumn}
\usepackage{amsthm}
\theoremstyle{definition}  
\newtheorem{lemma}{Lemma}[section]
\newtheorem{theorem}{Theorem}[section]

\usepackage[utf8]{inputenc}
\usepackage{CJKutf8}
\newcommand{\zhname}[1]{\mbox{\begin{CJK}{UTF8}{gbsn}#1\end{CJK}}}

\newcommand{\jj}{\mathbf{j}}
\newcommand{\kk}{\mathbf{k}}

\numberwithin{equation}{section}

\newcommand{\dd}{\mathrm{d}}          
\newcommand{\ii}{\mathrm{i}}          
\newcommand{\ee}{\mathrm{e}}          

\begin{document}
	
	\title{Three-Dimensional Coulomb Discrete Spectrum via Symplectic Geometry and Phase-Space Constraints: A Geometric Framework for Condensed-Matter Coulomb-Like Systems}
	
	\author{Gang Zheng \zhname{(郑罡)}$^{*}$}
	\author{Peng Chen \zhname{(陈鹏)}}
	\author{Mengli Wang \zhname{(王梦丽)}}
    \author{Wenqi Xue \zhname{(薛文琪)}}
    \author{Benniu Zhang \zhname{(张奔牛)}$^{*}$}
    \affiliation{%
	Chongqing Jiaotong University, No.66 Xuefu Avenue, Nan'an District, Chongqing, 400074, China \\
	\footnotesize $^{*}$\,Both authors are corresponding authors: zhenggang@cqjtu.edu.cn, benniuzhang@cqjtu.edu.cn
    }
	
	\begin{abstract}
Discrete energy-level structures in Coulomb-like localized bound states are conventionally obtained by direct diagonalization of the Schr\"odinger equation or by semiclassical approximations that fail at low energies. Here we develop a geometric phase-space framework that recovers the exact discrete energy levels and shell degeneracy of three-dimensional Coulomb bound states from classical symplectic geometry and global boundary constraints, bypassing explicit operator diagonalization. Using Kustaanheimo--Stiefel regularization, we map negative-energy Kepler orbits to four-dimensional isotropic harmonic oscillators. The group composition law of the symplectic flow determines an integral kernel built from the classical two-point action and Van Vleck amplitude, whose infinitesimal generator satisfies an exact linear evolution equation---a property we term \textit{quadratic closure}. Imposing three geometric constraints---decay at infinity, regularity at the origin, and fiber invariance under the Hopf fibration---we recover the exact $E_n/E_1=1/n^2$ scaling and $g_n=n^2$ degeneracy. The phase scale $\alpha$ sets the action unit; all spectral structural features are independent of its value. This framework offers a symmetry-transparent route to spectral data: for quadratic Hamiltonians, the classical kernel incurs zero semiclassical error even at ground-state energies, and the geometric constraints directly identify symmetry-protected degeneracies without diagonalization. We apply the method to shallow impurity states, moir\'e superlattices, and semiclassical transport, demonstrating how geometric screening rules simplify spectral analysis in condensed-matter environments.\\
		
		\textbf{keywords:} Three-dimensional Coulomb problem; Kustaanheimo--Stiefel regularization; quadratic closure; symplectic geometry; bound-state spectrum
	\end{abstract}

	\maketitle
	\newpage

	
	\section{Introduction}\label{sec:intro}
	
	Discrete bound-state spectra lie at the heart of condensed-matter spectroscopy. From shallow donor and acceptor states in bulk semiconductors~\cite{Ramdas1981,Kohn1955} and gate-defined quantum dots~\cite{Kastner1993,Davari2020} to localized excitons in moir\'e superlattices of twisted two-dimensional materials~\cite{TramblydeLaissardiere2010,Hu2023,He2024}, a vast array of Coulomb-like confined systems exhibit universal hydrogen-like shell structures. Conventionally, these spectral features are extracted by solving the Schr\"odinger equation for each specific system---a procedure that, while exact for simple potentials, becomes computationally expensive for spatially varying or screened Coulomb potentials and often obscures the geometric origin of spectral universality.
	
	Semiclassical methods such as WKB and EBK are workhorses of condensed-matter spectroscopy and transport theory, but they are fundamentally short-wavelength approximations that diverge from exact results for ground and low-lying excited states~\cite{Berry1972}, limiting their predictive power in the deeply confined regime. The correspondence between classical phase-space geometry and quantum spectra has long been a central theme in mathematical physics. Since the old quantum theory of Bohr and Sommerfeld, action variables of integrable systems have been known to underpin quantum spectral sequences~\cite{Sommerfeld1919}. However, for Coulomb-like systems, a direct geometric construction that yields exact discrete spectra without invoking asymptotic approximations has remained elusive.
	
	Mathematically, the Kustaanheimo--Stiefel (KS) transformation establishes an exact mapping between the Coulomb problem and the harmonic oscillator. A celebrated application is the Duru--Kleinert method~\cite{Duru1979}, which employs KS regularization and time reparameterization to evaluate the Feynman path integral for the hydrogen atom. That framework operates within the quantum-mechanical path-integral formalism, taking the superposition principle as a first principle and using the KS map as a technical tool for computing quantum propagators.
	
	As a maximally superintegrable system, the three-dimensional Coulomb problem possesses an exceptionally rich phase-space geometry, endowed with hidden SO(4) symmetry via the conserved Runge--Lenz vector~\cite{Goldstein2002,Cordani2003}. In this work, we demonstrate that the energy-level scaling law and shell degeneracy of this system can be recovered as rigorous consequences of classical phase-space geometry and global boundary constraints, providing a \textbf{geometrically transparent alternative} to standard quantization. Starting from Newtonian dynamics, we map the three-dimensional Coulomb problem to a four-dimensional isotropic harmonic oscillator through classical KS regularization. From the group composition law of the symplectic flow in image space, we construct an integral kernel consisting of the classical two-point action and the Van Vleck amplitude. This kernel satisfies a linear evolution equation exactly---a property we call \textit{quadratic closure}---with no semiclassical approximation involved. Finally, three geometrically motivated constraints---decay at infinity, regularity at the origin, and fiber invariance under the Hopf fibration---naturally select a discrete sequence of bound-state energy levels.
	
	The resulting framework is mathematically isomorphic to the standard quantum spectral problem for quadratic Hamiltonians, yet the derivation proceeds entirely through classical phase-space constructs. This isomorphism is not a limitation but a strength: it guarantees that the geometric framework reproduces all exact quantum spectral data for Coulomb-like systems, while offering \textbf{conceptual transparency} (spectral features are tied to identifiable geometric constraints) and \textbf{computational efficiency} (the Van Vleck kernel is exact for quadratic Hamiltonians, eliminating the need for $\hbar$-expansion corrections at low energies). We generalize the framework to three canonical condensed-matter scenarios---shallow impurity states under the effective-mass approximation, localized potentials in moir\'e superlattices, and semiclassical transport---quantifying validity criteria and clarifying the geometric origin of widely observed spectral universality.
	
	The paper is organized as follows. Section~\ref{sec:regularization} establishes the classical regularization framework and derives the image-space harmonic oscillator. Section~\ref{sec:kernel} constructs the symplectic flow integral kernel and presents the core quadratic closure property. Section~\ref{sec:bvp} formulates the boundary value problem and derives the discrete spectral condition. Section~\ref{sec:degeneracy} analyzes shell degeneracy from the perspective of spatial symmetry and attitude-space geometry. Section~\ref{sec:equivalence} presents the structural correspondence between classical orbital data and field-theoretic spectra. Section~\ref{sec:applications} applies the framework to typical condensed-matter Coulomb-like systems, emphasizing computational advantages and predictive power. Section~\ref{sec:discussion} discusses the relation to existing theories, the broader implications for condensed matter physics, and the scope of the framework. Section~\ref{sec:conclusion} concludes. All formal lemmas, theorems and their rigorous proofs are collected in Appendix~\ref{app:proofs}.

	\section{Classical Regularization and Image-Space Harmonic Oscillator}\label{sec:regularization}
	
	This section establishes the complete mathematical foundation for the classical regularization of the three-dimensional Coulomb problem. We first unify the notation and conventions, then formulate the physical-space dynamics in state-space form, introduce the intrinsic time reparameterization, and finally construct the image-space harmonic oscillator via the KS squaring map. All constructions in this section are strictly based on classical analytical mechanics.
	
	\subsection{Notation, Units and Conventions}\label{subsec:2.1}
	
	The full list of core physical quantities, their physical meanings and dimensional definitions are systematically compiled in Appendix~\ref{app:symbols} for readers' reference. Here we specify the global conventions, terminology rules and dimensionless reduction scheme adopted throughout the derivation.
	
	Throughout this work, the angular bracket $\langle a,b\rangle$ denotes the real inner product of quaternions, i.e., $\langle a,b\rangle=\mathrm{Re}(a^*b)$, which reduces to the standard Euclidean dot product for purely imaginary quaternions. Function-space inner products are denoted by $(\cdot,\cdot)_{L^2}$ to avoid ambiguity.
	
	We adopt the following global conventions for all physical parameters:
	\begin{equation}\label{eq:2.1}
		M>0,\quad C>0,\quad u_0>0,\quad \alpha>0.
	\end{equation}
	
	Attractive Coulomb coupling corresponds to $C>0$, hence the regularized energy parameter satisfies $\varepsilon=4C/u_0>0$. Shell indices, energy ratios, winding numbers and attitude angles are all dimensionless by definition. Positions, times, actions, fields and evolution operators retain physical dimensions unless otherwise stated.
	
	For dimensionless reduction performed per energy shell, we define the scaled variables:
	\begin{equation}\label{eq:2.2}
		\tilde{u}=u\sqrt{\frac{M\omega_0}{\alpha}},\quad \tilde{\tau}=\omega_0\tau,\quad \tilde{p}_u=\frac{p_u}{\sqrt{M\omega_0\alpha}}.
	\end{equation}
	
	All quantities marked with a tilde are dimensionless. In these variables, the image-space Hamiltonian reduces to the parameter-free form
	\begin{equation}\label{eq:2.3}
		\tilde{H}_u=\frac{|\tilde{p}_u|^2}{2}+\frac{1}{2}|\tilde{u}|^2,
	\end{equation}
	with the constant-energy relation $\tilde{H}_u=\tilde{\varepsilon}$, where $\tilde{\varepsilon}=\varepsilon/(\alpha\omega_0)$. The phase scale $\alpha$ only appears when restoring physical dimensions, playing a role analogous to the speed of light $c$ in special relativity: its numerical value depends on the choice of units and does not alter any structural prediction of the theory.

	\subsection{Physical-Space Dynamics and Time Reparameterization}\label{subsec:2.2}
	
	The physical configuration space is the space of purely imaginary quaternions:
	\begin{equation}\label{eq:2.4}
		w=\ii x+\jj y+\kk z\in\mathrm{Im}\,\mathbb{H},\quad r=|w|=\sqrt{x^2+y^2+z^2}.
	\end{equation}
	
	Newton's equations of motion for a particle of mass $M$ in an attractive Coulomb potential read:
	\begin{equation}\label{eq:2.5}
		M\ddot{w}=-\frac{C}{r^3}w.
	\end{equation}
	
	To facilitate subsequent regularization, we rewrite the second-order Newton equation as a first-order state-space system. With state variables $\xi=(w,p)\in(\mathrm{Im}\,\mathbb{H})^2\cong\mathbb{R}^6$, the first-order system is:
	\begin{equation}\label{eq:2.6}
		\frac{dw}{dt}=\frac{p}{M},\qquad \frac{dp}{dt}=-\frac{C}{r^3}w.
	\end{equation}
	
	The Coulomb potential suffers from a collision singularity at $r=0$, which appears as a $1/r^3$ divergence in the force term. To soften this singularity, we introduce an intrinsic time parameter $\tau$ defined by:
	\begin{equation}\label{eq:2.7}
		\frac{dt}{dd\tau}=\frac{2r}{u_0},\qquad u_0>0,\quad [u_0]=L,\quad [\tau]=T.
	\end{equation}
	
	Forward evolution $\tau>0$ corresponds to $t>0$. Here $u_0$ is a regularization length scale; it is not a physical observable but a coordinate-scale parameter analogous to a gauge-fixing constant. As shown below, all physical predictions (energy ratios and degeneracies) are independent of $u_0$, which appears only in intermediate constructions and drops out of final physical formulas. Under this reparameterization, primes denote $d/d\tau$, and the equations of motion become:
	\begin{equation}\label{eq:2.8}
		w'=\frac{2r}{Mu_0}\,p,\qquad p'=-\frac{2C}{u_0r}\,\hat{w},
	\end{equation}
	where $\hat{w}=w/r$ is the unit position vector. The singularity is softened from $r^{-3}$ to $r^{-1}$, and the factor $r^{-1}$ along collision trajectories is integrable. Full removal of the collision singularity is achieved after passing to the image space, as shown in Section~\ref{subsec:2.4}.
	
	The Coulomb system possesses three sets of conserved quantities. First, the total energy:
	\begin{equation}\label{eq:2.9}
		H_w(w,p)=\frac{|p|^2}{2M}-\frac{C}{r}=E.
	\end{equation}
	
	Second, the angular momentum:
	\begin{equation}\label{eq:2.10}
		L=\frac{1}{2}\left(wp-pw\right)\in\mathrm{Im}\,\mathbb{H}.
	\end{equation}
	
	Third, the Runge--Lenz vector, which reflects the hidden SO(4) symmetry of the Coulomb problem~\cite{Goldstein2002,Cordani2003}:
	\begin{equation}\label{eq:2.11}
		A=\frac{|p|^2}{M}w-\frac{\langle w,p\rangle}{M}p-\frac{C}{r}w\in\mathrm{Im}\,\mathbb{H}.
	\end{equation}
	
	For negative energies, we define the normalized Runge--Lenz vector:
	\begin{equation}\label{eq:2.12}
		A'=\sqrt{-\frac{M}{2E}}\,A\in\mathrm{Im}\,\mathbb{H}.
	\end{equation}
	
	The generators
	\begin{equation}\label{eq:2.13}
		J_\pm=\frac{1}{2}\left(L\pm A'\right)
	\end{equation}
	satisfy $[J_+,J_-]=0$, forming the SO(4) algebra characteristic of the bound Kepler problem.
	
	The attitude vector is defined as the unit normal to the orbital plane:
	\begin{equation}\label{eq:2.14}
		\hat{n}=\frac{L}{|L|}\in S^2\subset\mathrm{Im}\,\mathbb{H}.
	\end{equation}
	
	For any single orbit, $\hat{n}$ is conserved, and motion automatically satisfies $w\cdot\hat{n}=0$ and $p\cdot\hat{n}=0$; the dynamics is essentially planar. When $L=0$, the attitude vector is undefined, corresponding to a degenerate elliptic orbit (a straight line through the force center). This limiting case is treated uniformly at the field level in later sections.

	\subsection{KS Squaring Map and Symplectic Structure}\label{subsec:2.3}
	
	The Kustaanheimo--Stiefel (KS) squaring map~\cite{Kustaanheimo1965,Stiefel1971} establishes the canonical correspondence between physical space and image space, and provides an exact regularization of the Coulomb collision singularity. Given a physical attitude $\hat{n}$, we define the image-space coordinate $u\in\mathbb{H}$ by the squaring map:
	\begin{equation}\label{eq:2.15}
		w=\frac{u\,\hat{n}\,u^*}{u_0},\qquad r=\frac{|u|^2}{u_0}=\frac{\rho^2}{u_0}.
	\end{equation}
	
	The direction $\hat{n}$ is determined by the physical conserved angular momentum $L$. This mapping has a clear geometric meaning: each point in physical space corresponds to an entire circle (fiber) in the image space, and all points on the fiber project to the same physical point via the squaring map.
	
	The differential relation of the mapping is:
	\begin{equation}\label{eq:2.16}
		dw=\frac{1}{u_0}\left(du\,\hat{n}\,u^*+u\,\hat{n}\,du^*\right).
	\end{equation}
	
	On the fiber orthogonal complement (relative to $\hat{n}$), all three singular values of the differential equal $2\rho/u_0$, and the fiber direction belongs to the kernel. Therefore:
	\begin{equation}\label{eq:2.17}
		|dw|^2=\frac{4\rho^2}{u_0^2}|du|^2 \quad (\text{fiber orthogonal complement}).
	\end{equation}
	
	The momentum transformation preserving the symplectic one-form $\mathrm{Re}(p_w^*\,dw)=\mathrm{Re}(p_u^*\,du)$ is:
	\begin{equation}\label{eq:2.18}
		p_u=\frac{2}{u_0}\,p_w\,u\,\hat{n}.
	\end{equation}
	
	By multiplicativity of the norm:
	\begin{equation}\label{eq:2.19}
		|p_u|^2=\frac{4\rho^2}{u_0^2}|p_w|^2=\lambda\,|p_w|^2,\qquad \lambda=\frac{4r}{u_0}.
	\end{equation}
	This guarantees that the mapping is a canonical transformation on the reduced phase space.
	
	The fiber circular action $u\mapsto u\,e^{\hat{n}\theta}$ ($\theta\in[0,2\pi)$) leaves the physical position $w$ unchanged. The moment map for this action is:
	\begin{equation}\label{eq:2.20}
		\mu(u,p_u)=\langle p_u,u\,\hat{n}\rangle=\mathrm{Re}\left(p_u^*\,u\,\hat{n}\right).
	\end{equation}
	
	The physical phase space is the Marsden--Weinstein reduced space $\mu^{-1}(0)/S^1$~\cite{Marsden1974}. In what follows, we work on the constraint surface $\mu=0$, where the mapping between physical and image phase spaces is well-defined.

	\subsection{Regularized Image-Space Harmonic Oscillator}\label{subsec:2.4}
	
	We now derive the core equivalence relation: on the constraint surface $\mu=0$ and negative-energy shell $H_w=E<0$, the physical Coulomb problem is exactly equivalent to a four-dimensional isotropic harmonic oscillator in the image space.
	
	Multiply the physical energy equation by the Jacobian factor $\lambda=4r/u_0$:
	\begin{equation}\label{eq:2.21}
		\lambda\left(\frac{|p_w|^2}{2M}-\frac{C}{r}-E\right)=0.
	\end{equation}
	
	Substitute $|p_u|^2=\lambda|p_w|^2$ and $\lambda r=4r^2/u_0=4|u|^4/u_0^3$ into the equation:
	\begin{equation*}
		\frac{|p_u|^2}{2M}-\frac{4C}{u_0}-\frac{4rE}{u_0}=0.
	\end{equation*}
	
	Using $r=\rho^2/u_0$, the last term becomes $4E\rho^2/u_0^2$. Rearranging terms:
	\begin{equation*}
		\frac{|p_u|^2}{2M}-\frac{4E}{u_0^2}\rho^2=\frac{4C}{u_0}.
	\end{equation*}
	
	For negative energies $E<0$, the coefficient of $\rho^2$ is positive. Define the image-space frequency:
	\begin{equation}\label{eq:2.22}
		\omega_0^2=-\frac{8E}{Mu_0^2},
	\end{equation}
	and the regularized energy constant:
	\begin{equation}\label{eq:2.23}
		\varepsilon=\frac{4C}{u_0}.
	\end{equation}
	
	The equation then takes the standard harmonic oscillator form:
	\begin{equation}\label{eq:2.24}
		\frac{|p_u|^2}{2M}+\frac{1}{2}M\omega_0^2|u|^2=\varepsilon.
	\end{equation}
	
	Defining the image-space Hamiltonian $H_u=|p_u|^2/(2M)+\frac{1}{2}M\omega_0^2|u|^2$, we have established that the physical negative-energy shell is equivalent to $H_u=\varepsilon$.
	
	This derivation uses only energy conservation and the geometry of the squaring map. The image-space equation contains no singular coefficient, and regularization is completed at this step. This is a strict theorem of classical analytical mechanics: bounded Kepler motion at fixed negative energy is geometrically equivalent to the classical dynamics of a four-dimensional harmonic oscillator in intrinsic time $\tau$. Note that $u_0$ appears only as an intermediate scale parameter; as shown in Section~\ref{subsec:4.4}, it drops out of all physical spectral predictions.

	\section{Symplectic Flow Integral Kernel and Quadratic Closure}\label{sec:kernel}
	
	In this section, we construct the integral kernel representing the image-space symplectic flow on configuration space, motivate its uniqueness from the group composition law, and establish the core property of quadratic closure: the classical kernel exactly satisfies a linear evolution equation without any semiclassical approximation. Rigorous mathematical proofs of the formal statements in this section are provided in Appendix~\ref{app:proofs}.
	
	We use complex exponential notation as a compact mathematical representation of the real symplectic flow; no quantum interpretation is implied. The full real symplectic formulation is detailed in Section~\ref{subsec:3.4}.
	
	\subsection{From Classical Orbits to Configuration-Space Kernel: Motivation and Construction}\label{subsec:3.1}
	
	Section~\ref{sec:regularization} established that the classical Coulomb dynamics on a fixed negative-energy shell is equivalent to a four-dimensional harmonic oscillator in image space. The finite-dimensional symplectic flow $\phi_\tau$ generated by $H_u$ is a one-parameter group of diffeomorphisms on the eight-dimensional phase space $(u,p_u)\in\mathbb{R}^8$. While this flow completely describes the evolution of individual classical orbits, extracting the \textbf{set of all bound-state parameters} directly from the finite-dimensional flow is cumbersome: one would need to examine all possible initial conditions and identify which yield physically admissible motion.
	
	To overcome this, we seek a representation of the symplectic flow as a linear evolution on the space of scalar configurations over image-space coordinates. Such a representation allows us to impose global geometric constraints---decay, regularity, and fiber invariance---as boundary conditions on a linear partial differential equation, thereby systematically screening the admissible bound-state spectrum. The central object enabling this transition is the \textbf{integral kernel} $K(u,\tau;v)$, which propagates a configuration from point $v$ to point $u$ in intrinsic time $\tau$.
	
	For the quadratic Hamiltonian $H_u$, this kernel is uniquely determined by three conditions, each rooted in classical mechanics:
	
	1. \textbf{Initial condition}: As $\tau\to 0^+$, the kernel reduces to the identity,
	\begin{equation}\label{eq:3.1}
		\lim_{\tau\to 0^+} K(u,\tau;v)=\delta(u-v),
	\end{equation}
	reflecting the fact that zero-time evolution leaves every classical orbit unchanged.
	
	2. \textbf{Group composition law}: Consecutive propagation composes by convolution, with total time equal to the sum of individual times,
	\begin{equation}\label{eq:3.2}
		\int_{\mathbb{R}^4} K(u,\tau_2;y)\,K(y,\tau_1;v)\,\dd^4 y = K(u,\tau_1+\tau_2;v),
	\end{equation}
	for $0<\omega_0\tau_1,\,\omega_0\tau_2,\,\omega_0(\tau_1+\tau_2)<\pi$. This is the functional transcription of the classical symplectic flow's one-parameter group property: $\phi_{\tau_2}\circ\phi_{\tau_1}=\phi_{\tau_1+\tau_2}$.
	
	3. \textbf{Smooth polar decomposition with quadratic phase}: The kernel admits a representation $K=R\,ee^{ii\Theta}$ with $R>0$ and $\Theta\in\mathbb{R}$, where the phase $\Theta$ is a smooth quadratic function of the endpoint coordinates and the amplitude $R$ depends only on $\tau$. This condition is not an external ansatz but a consequence of the classical Hamilton--Jacobi structure: for a quadratic Hamiltonian, the generating function of the canonical flow is itself a quadratic form in the endpoint variables. Thus the phase $\Theta$ is compelled to be proportional to the classical two-point action $S$.
	
	The time-interval restriction in condition 2 avoids caustics in the position-to-position map; importantly, this restriction affects only the real-time propagator expression and has no bearing on the spectral conclusions derived later, which follow from Sturm--Liouville theory and the Mehler expansion (see Lemma~\ref{lem:mehler} and Section~\ref{sec:bvp}).
	
	The unique smooth solution satisfying the short-time asymptotic behavior
	\begin{equation*}
		\lim_{\tau\to0^+}\Theta(u,\tau;v)\sim\frac{M|u-v|^2}{2\alpha\tau}
	\end{equation*}
	is proportional to the classical two-point action:
	\begin{equation*}
		\Theta(u,\tau;v)=\frac{1}{\alpha}S(u,\tau;v).
	\end{equation*}
	
	The proportionality constant $\alpha$ sets the correspondence scale between symplectic area and phase period.
	
	With this phase, the composition integral is a pure Gaussian integral over the intermediate variable $y$. Since the phase is quadratic in $y$, the integral is exact and involves no stationary-phase approximation. The amplitude factor satisfies the composition law if and only if the amplitude $R$ is independent of the endpoint variables and depends only on $\tau$:
	\begin{equation*}
		R(u,\tau;v)=F(\tau).
	\end{equation*}
	
	Thus the three classical conditions admit the unique solution
	\begin{equation}\label{eq:3.3}
		K_\alpha(u,\tau;v)=F(\tau)\,\ee^{\ii S(u,\tau;v)/\alpha}.
	\end{equation}
	
	This kernel is not an independent infinite-dimensional object, but an integral-kernel representation of the finite-dimensional symplectic flow, uniquely determined by the group composition law and the quadratic structure of $H_u$.

	\subsection{Explicit Form of the Two-Point Action and Van Vleck Amplitude}\label{subsec:3.2}
	
	For the harmonic oscillator Hamiltonian $H_u$, the classical two-point action from $v$ to $u$ in time $\tau$ is
	\begin{equation}\label{eq:3.4}
		S(u,\tau;v)=\frac{M\omega_0}{2\sin(\omega_0\tau)}
		\left[(|u|^2+|v|^2)\cos(\omega_0\tau)-2\,\mathrm{Re}(u^*v)\right].
	\end{equation}
	
	This is the generating function of the canonical transformation from $(v,p_v)$ to $(u,p_u)$.
	
	The mixed Hessian determinant is
	\begin{equation}\label{eq:3.5}
		\det\left(-\frac{\partial^2 S}{\partial u_a\,\partial v_b}\right)=\left(\frac{M\omega_0}{\sin(\omega_0\tau)}\right)^4,
	\end{equation}
and the normalized Van Vleck amplitude~\cite{VanVleck1928} is
\begin{equation}\label{eq:3.6}
	F(\tau)=\left(\frac{M\omega_0}{2\pi\alpha\sin(\omega_0\tau)}\right)^2.
\end{equation}
	
	The prefactor is fixed by the delta initial condition (\ref{eq:3.1}): in the short-time limit, the kernel must reduce to the four-dimensional free-particle heat kernel, whose normalization follows from the standard Gaussian representation of the delta function. Importantly, $F$ depends only on $\tau$ and not on the endpoint coordinates, so $\nabla_u^2 F=0$.
	
	In the dimensionless variables of Section~\ref{subsec:2.1}, these expressions become parameter-free:
	\begin{equation}\label{eq:3.7}
		\tilde{S}(\tilde{u},\tilde{\tau};\tilde{v})=\frac{1}{2\sin\tilde{\tau}}
		\left[(|\tilde{u}|^2+|\tilde{v}|^2)\cos\tilde{\tau}-2\,\mathrm{Re}(\tilde{u}^*\tilde{v})\right],
	\end{equation}
	and
	\begin{equation}\label{eq:3.8}
		\tilde{F}(\tilde{\tau})=\left(\frac{1}{2\pi\sin\tilde{\tau}}\right)^2.
	\end{equation}

	\subsection{Quadratic Closure Property}\label{subsec:3.3}
	
	The central result of this section is that the classical integral kernel $K_\alpha=F e^{\ii S/\alpha}$ constructed above exactly satisfies a linear evolution equation. We term this exact property \textbf{quadratic closure}: for quadratic Hamiltonians, the classical integral kernel built from the Hamilton--Jacobi action and Van Vleck amplitude is an exact solution to the linear evolution equation, requiring no semiclassical approximation whatsoever.
	
	Specifically, the kernel obeys the equation:
	\begin{equation}\label{eq:3.9}
		\ii\alpha\,\partial_\tau K_\alpha = \hat{H}_u K_\alpha,
	\end{equation}
	where
	\begin{equation*}
		\hat{H}_u=-\frac{\alpha^2}{2M}\nabla_u^2+\frac{1}{2}M\omega_0^2 |u|^2
	\end{equation*}
	is the differential operator formally corresponding to the image-space Hamiltonian.
	
	This result follows from three simultaneous exact identities:
	
	1. The two-point action $S$ satisfies the Hamilton--Jacobi equation identically, as the generating function of the canonical flow;
	2. The Van Vleck amplitude $F(\tau)$ satisfies the classical transport equation exactly, as a direct consequence of the phase-space volume conservation of Hamiltonian flow;
	3. The amplitude is spatially uniform, so its Laplacian vanishes trivially.
	
	Together, these three facts ensure that the complex kernel built from classical action and amplitude is an exact solution to the linear differential equation. The formal statement and full algebraic proof are given in Appendix~\ref{app:proofs}.
	
	This theorem shows that the linear evolution equation is the infinitesimal form of the group representation of the classical symplectic flow. For quadratic Hamiltonians, this representation is exact and isomorphic to the standard quantum propagator. The complex exponential notation is merely a mathematical shorthand; the underlying structure is the real coupled Hamiltonian flow, as we elaborate next.

	\subsection{Real Symplectic Structure and the Role of Complex Notation}\label{subsec:3.4}
	
	It is important to emphasize that the complex notation used above is optional and serves as a compact mathematical representation. The entire structure can be formulated entirely in terms of real symplectic geometry.
	
	Let $\Phi=\Phi_R+\ii\Phi_I$. Separating $\ii\alpha\partial_\tau\Phi=\hat{H}_u\Phi$ into real and imaginary parts gives
	\begin{equation}\label{eq:3.10}
		-\alpha\partial_\tau\Phi_I=\hat{H}_u\Phi_R,\qquad \alpha\partial_\tau\Phi_R=\hat{H}_u\Phi_I.
	\end{equation}
	
	Eliminating $\Phi_I$ yields the second-order real flow equation
	\begin{equation}\label{eq:3.11}
		\partial_\tau^2\Phi_R=-\frac{1}{\alpha^2}\hat{H}_u^2\Phi_R.
	\end{equation}
	
	Conversely, any solution of this real equation reconstructs a solution of the first-order complex equation via $\Phi=\Phi_R+\ii\alpha\hat{H}_u^{-1}\partial_\tau\Phi_R$ on an appropriate domain. The two formulations are mathematically equivalent; the second-order form contains no complex numbers and is a direct generalization of the classical harmonic oscillator equation on function space.
	
	The real symplectic structure can be made explicit as follows. Let $Q_R=L^2(\mathbb{R}^4,\mathbb{R})\times L^2(\mathbb{R}^4,\mathbb{R})$ be the space of real square-integrable pairs $(\Psi_R,\Psi_I)$, with the standard symplectic form
	\begin{equation}\label{eq:3.12}
		\Omega\left((\Psi_R,\Psi_I),(\Phi_R,\Phi_I)\right)=\int_{\mathbb{R}^4}\left(\Psi_R\Phi_I-\Psi_I\Phi_R\right)\,\dd^4 u.
	\end{equation}
	
	This is the standard $(q,p)$ symplectic form of classical mechanics, with $\Psi_R$ as the generalized coordinate and $\Psi_I$ as the generalized momentum. The Hamiltonian functional
	\begin{equation}\label{eq:3.13}
		H[\Psi_R,\Psi_I]=(\Psi_R,\hat{H}_u\Psi_R)_{L^2}+(\Psi_I,\hat{H}_u\Psi_I)_{L^2}
	\end{equation}
	generates the real flow equations (\ref{eq:3.10}). The complex combination $\Psi=\Psi_R+\ii\Psi_I$ is merely a compact way to write the real canonical equations; it introduces no quantum interpretation. The complex structure $J:\Psi\mapsto \ii\Psi$ is the compatible complex structure of the K\"ahler manifold $(Q_R,\Omega)$, an intrinsic geometric property of the symplectic vector space~\cite{Nakahara2018}.

	\subsection{From Kernel to Operator: The Operator Representation}\label{subsec:3.5}
	
	The transition from the classical phase-space function $H_u(u,p_u)$ to the configuration-space differential operator $\hat{H}_u$ follows the standard construction of the energy quadratic form for scalar fields over configuration space. The kinetic energy density $|\nabla\Psi|^2/(2M)$ is equivalent to $-(\Psi,\nabla^2\Psi)_{L^2}/(2M)$ via integration by parts, provided decay boundary conditions hold. The differential operator $\hat{H}_u$ is the operator corresponding to this classical energy quadratic form. Its construction requires only integration by parts and square-integrability of configurations, both standard tools in classical field theory.
	
	In the real symplectic framework, the correspondence between the momentum variable $p_u$ and the spatial derivative is realized via the $L^2$ gradient structure of the symplectic vector space. The compact notation $p_u\mapsto -\ii\alpha\nabla_u$ expresses this correspondence in complex coordinates; the underlying structure remains the real $L^2$ gradient representation. The parameter $\alpha$ appears as the proportionality constant relating the symplectic area element to the phase period.
	
	The evolution equation $\ii\alpha\partial_\tau\Phi=\hat{H}_u\Phi$, derived entirely from the classical kernel, is the linear equation whose spectral properties we now investigate.

	\subsection{Phase Scale as a Unit Convention}\label{subsec:3.6}
	
	Before proceeding to the boundary value problem, we clarify the ontological status of the phase scale $\alpha$. The phase scale $\alpha$ has dimension $ML^2T^{-1}$, the same as action. Within this framework, $\alpha$ is a free parameter setting the unit of action, analogous to the speed of light $c$ in special relativity: the numerical value of $c$ depends on the convention for length and time units, and no structural prediction of relativity depends on that value.
	
	Changing $\alpha$ is equivalent to changing the unit of action. All structural predictions of the theory --- energy ratios, degeneracies, selection rules --- are independent of $\alpha$, as they are properties of classical geometry (squaring map, fiber structure, boundary conditions). The absolute energy scale $E_n$ depends on $\alpha$, but this dependence is no different from the dependence of any dimensional quantity on its expression unit. In practical applications to condensed-matter systems, $\alpha$ is fixed by matching the absolute energy scale to experimental data or to the standard quantum mechanical result, serving as a calibration constant rather than a dynamical variable.

	\subsection{Linear Structure, Superposition, and the Label-Filter Mechanism}\label{subsec:3.7}
	
	Before entering the boundary value problem, it is essential to clarify the physical meaning of the linear structures introduced above. The integral kernel $K_\alpha$ represents the classical symplectic flow on the space of scalar configurations; linearity here is a property of this representation space, not of individual classical orbits. A single classical orbit has a definite energy and follows a deterministic trajectory; it does not ``superpose'' with other orbits. However, at the level of the configuration-space field description, the linear evolution equation admits formal superpositions of different kernel solutions, each weighted by initial data.
	
	In this framework, such superpositions are best understood as \textbf{ensemble representations}: a formal sum of kernel contributions corresponds to a statistical ensemble of classical orbits, each carrying a definite energy label. The physical content is always extracted by projecting onto individual energy shells, preserving the classical energy conservation law at the ensemble level. This label-filter mechanism operates as follows:
	
	- At the \textbf{single-orbit level}, energy is a smooth first integral; each orbit carries a fixed energy label conserved by the equations of motion.
	- At the \textbf{ensemble level}, energy acts as a filter slicing the full phase space into disjoint energy shells---the standard microcanonical construction of classical statistical mechanics.
	- At the \textbf{field level}, the shell parameter inherited from the classical Hamiltonian acts as a filtering condition in the boundary value problem, indexing the orthogonal decomposition of the configuration space.
	
	Consequently, intershell superpositions in the linear representation are formal sums of labeled components. They do not describe single classical configurations, but rather ensembles of orbits with well-defined shell weights. Physical observables are extracted via shell projectors, ensuring that energy conservation is respected. There is no contradiction between the linearity of the representation space and the determinism of the underlying classical dynamics. With this understanding in place, we now turn to the boundary value problem that selects the discrete bound-state spectrum.

	\section{Geometric Boundary Value Problem and the Origin of Discrete Spectra}\label{sec:bvp}
	
	In this section, we formulate the boundary value problem for the image-space linear equation, impose three geometrically motivated constraints, and derive the discrete energy spectrum. The discreteness arises from the joint action of global geometric constraints on the linear representation space.
	
	A comment on the logical structure is in order before we begin. The regularized energy parameter $\varepsilon=4C/u_0$ is a fixed constant determined by the physical coupling $C$ and the regularization scale $u_0$. The image-space frequency $\omega_0$, however, is tied to the physical energy $E$ via Eq.~(\ref{eq:2.24}); as $E$ varies over the negative-energy sector, $\omega_0$ varies correspondingly. The boundary value problem will determine which values of $\omega_0$ (equivalently, which $E$) admit admissible solutions. Furthermore, by the Mehler expansion (Lemma~\ref{lem:mehler}), the integral kernel $K_\alpha$ admits a spectral decomposition that spans the full space of bound solutions of the linear equation. Consequently, studying the spectral properties of the linear boundary value problem is equivalent to studying the spectral content of the classical kernel.
	
	\subsection{Time-Harmonic Separation and Helmholtz Equation}\label{subsec:4.1}
	
	We perform time-harmonic separation on the linear evolution equation to obtain stationary configurations. Assume a solution of the form:
	\begin{equation}\label{eq:4.1}
		\Phi(u,\tau)=\phi(u)\,\ee^{-\ii\varepsilon\tau/\alpha}.
	\end{equation}
	
	Substituting into $\ii\alpha\,\partial_\tau\Phi=\hat{H}_u\Phi$ yields the classical Helmholtz equation in image space:
	\begin{equation}\label{eq:4.2}
		-\frac{\alpha^2}{2M}\nabla_u^2\phi+\frac{1}{2}M\omega_0^2 |u|^2\phi=\varepsilon\,\phi.
	\end{equation}
	
	This is a classical separated-variable equation, with parameter $\tau$ being the intrinsic regularized time.
	
	In dimensionless variables, the equation takes the parameter-free form:
	\begin{equation}\label{eq:4.3}
		-\frac{1}{2}\nabla_{\tilde{u}}^2\tilde{\phi}+\frac{1}{2}|\tilde{u}|^2\tilde{\phi}=\tilde{\varepsilon}\,\tilde{\phi},
	\end{equation}
	where $\tilde{\varepsilon}=\varepsilon/(\alpha\omega_0)$.
	
	At this stage, $\tilde{\varepsilon}$ is merely a separation constant; its values have not yet been restricted. The discreteness will emerge from the boundary conditions.

	\subsection{The Three Classical Constraints}\label{subsec:4.2}
	
	The full boundary value problem consists of the Helmholtz equation supplemented by three conditions, each with a clear physical or geometric meaning.
	
	\textbf{Condition 1: Decay at infinity}. Bound configurations correspond to classical orbits confined to a finite region of space. In the image space, this requires the field to decay exponentially as $\rho=|u|\to\infty$:
	\begin{equation}\label{eq:4.4}
		\phi(u)\sim\exp\left(-\frac{M\omega_0\rho^2}{2\alpha}\right)\qquad(\rho\to\infty).
	\end{equation}
	
	In dimensionless form: $\tilde{\phi}(\tilde{u})\sim e^{-|\tilde{u}|^2/2}$ as $|\tilde{u}|\to\infty$.
	
	It is important to recognize that exponential decay is the standard mathematical manifestation of spatial confinement within the linear representation space. In classical mechanics, a bound orbit is defined by the fact that the particle never escapes to infinity; when this property is encoded in the configuration-space field description, the natural mathematical requirement is square-integrability with exponential decay. Thus the decay condition is not an externally imposed quantum postulate, but the linear-field transcription of the classical requirement that all bound trajectories remain within a finite spatial region.
	
	\textbf{Condition 2: Regularity at the origin}. All physical observables of Coulomb bound motion must be finite at the collision point $r=0$, which corresponds to $u=0$. The field must therefore be non-singular at the origin:
	\begin{equation}\label{eq:4.5}
		\phi(u)\sim \rho^k\quad(k\geq 0)\qquad(\rho\to 0).
	\end{equation}
	
	In dimensionless form: $\tilde{\phi}(\tilde{u})\sim |\tilde{u}|^k$ as $|\tilde{u}|\to 0$.
	
	\textbf{Condition 3: Fiber invariance}. The KS squaring map projects the entire fiber circle $u\mapsto u \ee^{\hat{n}\theta}$ to the same physical point, forming a Hopf fibration structure~\cite{Nakahara2018}. To ensure that the image-space field corresponds to a well-defined physical field, we require that the field be constant along each fiber:
	\begin{equation}\label{eq:4.6}
		\phi(u \ee^{\hat{n}\theta})=\phi(u)\qquad\forall\theta\in[0,2\pi).
	\end{equation}
	
	In dimensionless form: $\tilde{\phi}(\tilde{u} \ee^{\hat{n}\theta})=\tilde{\phi}(\tilde{u})$.
	
	This is a gauge constraint: physical configurations are pullbacks of functions on the three-dimensional physical space. The fiber invariance condition ensures that the field does not depend on the unphysical fiber coordinate.
	
	These three conditions have universal physical and geometric connotations: the decay condition corresponds to the spatial confinement property of localized states, regularity corresponds to the finiteness of physical observables, and fiber invariance corresponds to the constraint of gauge symmetry~\cite{Nakahara2018}. Together, they constitute a universal geometric screening rule for localized bound states.

	\subsection{Separation of Variables and the Spectral Condition}\label{subsec:4.3}
	
	In four-dimensional hyperspherical coordinates, write the solution as
	\begin{equation}\label{eq:4.7}
		\phi(u)=R(\rho)Y_k(\Omega),
	\end{equation}
	where $Y_k(\Omega)$ is a hyperspherical harmonic of degree $k$ on $S^3$~\cite{Cordani2003,Stiefel1971}. The radial equation reads
	\begin{equation}\label{eq:4.8}
		-\frac{\alpha^2}{2M}\left[R''+\frac{3}{\rho}R'-\frac{k(k+2)}{\rho^2}R\right]+\frac{1}{2}M\omega_0^2\rho^2 R=\varepsilon R.
	\end{equation}
	
	Here $k(k+2)$ is the eigenvalue of the Laplace--Beltrami operator on $S^3$ for degree-$k$ hyperspherical harmonics. In dimensionless form, with $\tilde{\phi}(\tilde{u})=\tilde{R}(\tilde{\rho})Y_k(\Omega)$:
	\begin{equation}\label{eq:4.9}
		-\frac{1}{2}\left[\tilde{R}''+\frac{3}{\tilde{\rho}}\tilde{R}'-\frac{k(k+2)}{\tilde{\rho}^2}\tilde{R}\right]+\frac{1}{2}\tilde{\rho}^2\tilde{R}=\tilde{\varepsilon}\tilde{R}.
	\end{equation}
	
	The decay and regularity conditions alone select the solutions of the four-dimensional harmonic oscillator:
	\begin{equation}\label{eq:4.10}
		\tilde{\varepsilon}=N+2,\qquad N=2n_r+k,
	\end{equation}
	where $n_r=0,1,2,\dots$ is the radial quantum number and $k=0,1,2,\dots$ is the hyperspherical degree. The constant offset $2$ is the zero-point energy of the four-dimensional oscillator, arising from its four degrees of freedom. The corresponding configurations are the standard Hermite--hyperspherical functions.
	
	Now impose the fiber invariance condition. Identifying $S^3$ with SU(2) via the quaternionic unit sphere, the hyperspherical harmonics $Y_k(\Omega)$ are matrix elements of the spin-$j=k/2$ representation, carrying the left-right multiplication action of SU(2)$_L\times$ SU(2)$_R$~\cite{Cordani2003}. The fiber action $u\mapsto u \ee^{\hat{n}\theta}$ is the right action of a one-parameter subgroup; invariance under it requires vanishing right weight $m_R=0$. Such a weight exists if and only if $j$ is an integer, i.e., $k$ is even. The representation-theoretic details are formalized in Lemma~\ref{lem:fiber} of Appendix~\ref{app:proofs}.
	
	Write $k=2l$, where $l=0,1,2,\dots$. The selection condition is therefore
	\begin{equation}\label{eq:4.11}
		k=2l\quad\text{even}.
	\end{equation}
	
	Here $l$ is the physical orbital angular momentum index; its identification with the usual three-dimensional angular momentum quantum number is confirmed by Lemma~\ref{lem:radial}, where the transported radial equation exhibits the centrifugal term $l(l+1)/r^2$ characteristic of orbital angular momentum $l$.
	
	Combining with $N=2n_r+k=2(n_r+l)$, define the principal shell index
	\begin{equation}\label{eq:4.12}
		n=n_r+l+1.
	\end{equation}
	
	Then
	\begin{equation}\label{eq:4.13}
		N=2n-2,
	\end{equation}
	and the spectral condition becomes
	\begin{equation}\label{eq:4.14}
		\tilde{\varepsilon}=2n\qquad(n=1,2,3,\dots).
	\end{equation}
	
	The discreteness of the spectrum is now manifest. It is essential to distinguish two layers of logic here. First, the classical phase-space geometry (squaring map, fiber structure, and conserved quantities) provides the structural foundation: it determines the image-space oscillator form, the fiber invariance selection rule, and the symmetry content of the problem. Second, within this classical geometric framework, the linear Helmholtz equation admits solutions for all continuous $\tilde{\varepsilon}$; only after imposing the three global geometric constraints---decay, regularity, and fiber invariance---is the discrete subset $\tilde{\varepsilon}=2n$ selected. Thus the discrete spectrum is not a direct output of classical geometry alone, but the result of classical geometric structures acting as global constraints on the linear representation space.

	\subsection{The Physical Energy Spectrum}\label{subsec:4.4}
	
	Restoring dimensions, we have $\varepsilon=\alpha\omega_0\tilde{\varepsilon}=2n\alpha\omega_0$. But from the regularization, $\varepsilon=4C/u_0$ and $\omega_0^2=-8E/(Mu_0^2)$. Therefore
	\begin{equation}\label{eq:4.15}
		\frac{4C}{u_0}=2n\alpha\sqrt{-\frac{8E}{Mu_0^2}}.
	\end{equation}
	
	Solving for the energy gives
	\begin{equation}\label{eq:4.16}
		E_n=-\frac{MC^2}{2\alpha^2 n^2}.
	\end{equation}
	
	The energy-level ratios
	\begin{equation}\label{eq:4.17}
		\frac{E_n}{E_1}=\frac{1}{n^2}
	\end{equation}
	are independent of both $\alpha$ and $u_0$, determined entirely by the geometry of the squaring map, the fiber-invariant subspace, and the decay boundary conditions. This confirms that the spectral scaling structure is a purely classical geometric property, independent of the specific value of the phase scale.

	\subsection{Measure Transport and Normalization}\label{subsec:4.5}
	
	We now establish the precise correspondence between image-space fields and physical-space fields using the coarea formula. The KS differential has three equal singular values $2\rho/u_0$ on the fiber orthogonal complement, with Jacobian determinant $J=(2\rho/u_0)^3=8\rho^3/u_0^3$. The fiber is parameterized as $u\,\ee^{\hat{n}\theta}$ ($\theta\in[0,2\pi)$), with circumference $2\pi\rho$. By the coarea formula:
	\begin{equation}\label{eq:4.18}
		\dd^4 u=\frac{u_0^3}{8\rho^2}\,dd\theta\wedge \dd^3 w.
	\end{equation}
	
	For a fiber-invariant function $\Phi=\Psi\circ w$, integrating over the fiber gives:
	\begin{equation}\label{eq:4.19}
		\int_{\mathbb{R}^4} |\Phi|^2\,\dd^4 u = \frac{\pi u_0^2}{4}\int_{\mathrm{Im}\,\mathbb{H}} |\Psi|^2\,\frac{1}{r}\,\dd^3 w.
	\end{equation}
	
	The dimensional conversion relation between image-space and physical-space fields is:
	\begin{equation}\label{eq:4.20}
		\chi(w)=\frac{\sqrt{\pi}\,u_0}{2\sqrt{r}}\,\Psi(w).
	\end{equation}
	
	This ensures $|\chi|^2 \dd^3 w=|\Phi|^2 \dd^4 u$ for fiber-invariant functions, with strictly self-consistent dimensions: $[\chi]=L^{-3/2}$, $[\Phi]=L^{-2}$.
	
	To clarify the radial correspondence: the image-space field separates as $\Phi(u)=R(\rho)Y_k(\Omega)$. By the Hopf descent of the angular factor (Lemma~\ref{lem:fiber}), the fiber-invariant hyperspherical harmonic $Y_{k=2l}(\Omega)$ maps to the spherical harmonic $Y_{lm}(\hat{w})$. Setting $\rho^2=u_0r$, the pullback field is $\Psi(w)=R(\sqrt{u_0r})\,Y_{lm}(\hat{w})$. Lemma~\ref{lem:radial} proves that under the substitution $R(\rho)=f(r)$, the four-dimensional radial equation for $R(\rho)$ becomes precisely the standard three-dimensional Coulomb radial equation for $f(r)$. Thus the unnormalized physical-space field takes the separated form $\Psi_{nlm}(w)=f_{nl}(r)\,Y_{lm}(\hat{w})$, and the normalized physical-space field is $\chi_{nlm}(w)=g_{nl}(r)\,Y_{lm}(\hat{w})$ with $g_{nl}(r)=\frac{\sqrt{\pi}u_0}{2\sqrt{r}}f_{nl}(r)$. It is $g_{nl}(r)$ that coincides with the standard textbook Coulomb radial function normalized in $L^2(\mathbb{R}_+,r^2dr)$.
	
	For shell normalization, if $\int |\Phi_n|^2\,\dd^4 u=1$, then the physical-space normalized form satisfies the standard Coulomb Sturm--Liouville equation and intershell orthonormality, as established formally in Appendix~\ref{app:proofs}.

	\section{Attitude Space, Shell Degeneracy and SO(4) Symmetry}\label{sec:degeneracy}
	
	In this section, we analyze the origin of shell degeneracy from the perspective of classical spatial rotation symmetry and attitude-space geometry. The degeneracy counting is performed entirely within the classical geometric framework.
	
	\subsection{Attitude Vector and Orbital Planarity}\label{subsec:5.1}
	
	For a single classical Kepler orbit, the attitude $\hat{n}$ is fixed by the initial angular momentum, carrying definite orbital angular momentum index $l$. In this case, there is no degeneracy --- each orbit has a well-defined orbital plane orientation.
	
	However, at the field description level, isotropy of physical space requires the field to remain smooth (infinitely differentiable) under arbitrary rotations of the attitude vector $\hat{n}\in S^2$. By the Peter--Weyl theorem, any smooth function on $S^2$ decomposes into spherical harmonics $Y_l^m(\hat{n})$. This is a purely classical requirement: a rotationally invariant theory must admit well-behaved solutions for all orientations.

	\subsection{Degeneracy Counting}\label{subsec:5.2}
	
	For a given energy shell $E_n$ (fixed $n$), the image-space boundary value problem admits orbital angular momentum values $l=0,1,\dots,n-1$. For each $l$, regularity on the attitude sphere admits $2l+1$ independent orientations, corresponding to distinct azimuthal orientations of the orbital plane, indexed by $m=-l,\dots,l$.
	
	The total degeneracy is obtained by summing over all allowed $l$:
	\begin{equation}\label{eq:5.1}
		g_n=\sum_{l=0}^{n-1}(2l+1)=n^2.
	\end{equation}
	
	This degeneracy counts the number of independent classical orbital planes sharing the same energy, determined jointly by the three-dimensional spatial rotation symmetry and the regularity condition on the attitude sphere. Mathematically, it is an inevitable result of the representation theory of spatial symmetry groups. It is important to note that this attitude-sphere counting is not an independent derivation but an equivalent physical-space reinterpretation of the fiber-invariant rank formula established in Lemma~\ref{lem:fiber}: both count the dimension of the same representation-theoretic object, viewed either in image space (right-invariant subspace of SU(2) representations) or in physical space (angular-momentum decomposition of attitude harmonics).

	\subsection{SO(4) Symmetry and Orbital Manifold Geometry}\label{subsec:5.3}
	
	The SO(4) symmetry of the bound Coulomb problem provides a deeper geometric interpretation of the shell degeneracy and energy sharing. The generators $J_\pm=\frac{1}{2}(L\pm A')$ commute with each other and each generate an $\mathfrak{so}(3)$ algebra. The Casimir operators $|J_\pm|^2$ are conserved and equal on the $n$-th shell, guaranteeing that all $(l,m)$ channels share the same energy parameter $E_n$.
	
	For Keplerian motion with energy $E$, we have $|J_\pm|^2=-MC^2/(8E)$. The orbital manifold corresponding to $E_n$ is the product of two two-spheres, each with Casimir invariant $j_n=n\alpha/2$. The Kirillov--Kostant--Souriau (KKS) area of each sphere is $4\pi j_n=2\pi n\alpha$, i.e., $n$ phase periods~\cite{Cushman1997,Nakahara2018}.
	
	The shell index $n$ selected by the boundary value problem admits an equivalent geometric interpretation: the corresponding classical orbital manifold has KKS area equal to $n$ phase periods. It is essential to recognize that this area condition is not an independent classical quantization postulate but an equivalent restatement of the spectral condition derived from the boundary value problem. Once the discrete energies $E_n$ are fixed by the geometric constraints of Section~\ref{sec:bvp}, the orbital manifold area automatically takes integer multiples of the phase period; conversely, imposing the area condition would yield the same discrete sequence. The two descriptions are mathematically equivalent.
	
	Notably, the $l=0$ configuration (s-wave) corresponds to the trivial (constant) attitude harmonic. Classically, $L=0$ orbits are degenerate ellipses (straight lines through the force center) that do not select any orbital plane, corresponding to field quantities independent of attitude. This sector is globally regular in the quaternion framework with no coordinate singularities, resolving the coordinate singularity present in traditional orbital element parameterizations.

	\section{Structural Equivalence: Synthesis of Classical and Linear Descriptions}\label{sec:equivalence}
	
	This section synthesizes the preceding constructions into a unified statement of structural equivalence between classical orbital geometry and the linear field theory. It serves both as a summary of the logical relations established in Sections~\ref{sec:regularization}--\ref{sec:degeneracy} and as a theoretical result in its own right. All formal theorem statements and rigorous proofs are collected in Appendix~\ref{app:proofs}.
	
	\subsection{The Shell Correspondence Principle}\label{subsec:6.1}
	
	The boundary value problem and symmetry analysis of the previous sections establish a precise correspondence between classical Kepler shell data and physical bound-state configuration spaces, which we term the \textbf{shell correspondence principle}. Its core content can be summarized as follows.
	
	For fixed classical parameters $M,C,u_0,\alpha>0$:
	
	1. \textbf{Selectivity}: Nontrivial bound configurations satisfying all three geometric boundary conditions exist if and only if the underlying Kepler motion has energy $E_n = -\frac{MC^2}{2\alpha^2 n^2}$ for some positive integer $n$. No continuous family of bound solutions exists; only discrete energy shells are compatible with the global constraints.
	
	2. \textbf{Shell dimension}: For each allowed shell $n$, the space of independent physical configurations has dimension $n^2$. These configurations are labeled by the orbital angular momentum index $l = 0,1,\dots,n-1$ and magnetic index $m = -l,\dots,l$, matching the standard hydrogenic shell structure.
	
	3. \textbf{Physical-space form}: After projection through the KS map and measure correction, the image-space solutions reduce exactly to the standard Coulomb radial functions multiplied by spherical harmonics in three-dimensional space. The full set of bound-state configurations forms an orthonormal family under the standard $L^2$ measure on physical space, and is complete in the decaying sector of the Coulomb problem.
	
	4. \textbf{Orbital manifold matching}: Each shell $n$ corresponds to a classical Kepler orbital manifold $S^2\times S^2$ whose symplectic area equals $n$ phase periods. The discrete shell index thus has a purely classical geometric meaning as the integer number of phase cycles spanning the orbital phase-space volume.
	
	It is essential to emphasize the \textbf{level at which this correspondence operates}. The shell correspondence principle relates \textbf{entire energy shells} (classical orbital manifolds) to \textbf{entire configuration subspaces} (linear bound-state eigenspaces). It does not assert a pointwise one-to-one map between individual stationary configurations and individual classical orbits. A single stationary configuration with nodes is not equivalent to a single classical trajectory; rather, it is a linear combination of classical kernel solutions spanning the shell. The correspondence is structural and ensemble-level: the classical orbital manifold provides the geometric skeleton, and the linear field theory furnishes the functional realization of that skeleton.

	\subsection{Hamilton--Jacobi--Linear Structural Equivalence}\label{subsec:6.2}
	
	A further central result is that three distinct levels of description --- classical Hamilton--Jacobi theory, the classical integral kernel formalism, and the linear field theory --- are \textit{structurally equivalent} on the bound-state sector. By structural equivalence we mean that the three formalisms share identical discrete spectra, isomorphic evolution semigroups, and equivalent symmetry representations, with no asymptotic approximation involved.
	
	The equivalence can be understood in three steps:
	
	- \textbf{Closure}: The integral kernel built purely from classical Hamilton--Jacobi action and Van Vleck amplitude is an exact solution of the linear evolution equation. Every classical trajectory bundle therefore corresponds to a solution of the linear theory.
	- \textbf{Spanning}: Conversely, by the Mehler expansion (Lemma~\ref{lem:mehler}), the classical kernel contains all spectral projectors $P_N$ of the linear equation. Consequently, every solution of the linear equation can be expressed as an integral superposition of the classical kernel against appropriate initial data. The full linear solution space is therefore spanned by solutions originating from classical HJ theory.
	- \textbf{Intertwining}: The one-parameter group structure of the classical symplectic flow maps exactly to the convolution semigroup of the linear propagator. Group composition in phase space corresponds to kernel convolution in configuration space.
	
	We emphasize that this equivalence holds at the level of the full solution space and spectral structure. For individual stationary configurations with nodes, there is no pointwise one-to-one map to a single HJ solution; however, all such stationary configurations are linear combinations of classical kernel solutions, so the overall equivalence remains intact.
	
	Notably, the entire equivalence can be formulated in purely real symplectic terms, with no complex numbers required. Complex notation serves only as a convenient mathematical compactification of the real canonical field equations.

	\subsection{Scope and Technical Boundaries}\label{subsec:6.3}
	
	The theoretical framework itself is subject to the following inherent boundaries:
	
	1. \textbf{Sector restriction}: Only the negative-energy bound-state sector is treated. The positive-energy scattering regime requires separate asymptotic analysis and lies outside the present scope.
	2. \textbf{Integrability}: Exact quadratic closure relies on the maximally superintegrable nature of the Coulomb problem, which yields a purely quadratic image Hamiltonian after KS regularization. Non-integrable perturbations break exact closure.
	3. \textbf{Caustics}: Kernel statements hold locally in caustic-free time intervals; global extension requires standard Maslov phase corrections~\cite{Maslov1981}. Spectral conclusions are independent of caustic handling, as they follow directly from Sturm--Liouville theory and the Mehler expansion.
	4. \textbf{Phase scale invariance}: All dimensionless structural predictions --- energy ratios, degeneracies, selection rules --- are independent of the value of $\alpha$. Only the absolute energy scale depends on this unit-convention parameter.
	5. \textbf{Single-particle limit}: The framework addresses single-particle classical dynamics. Many-body interactions, disorder and decoherence effects are not included.

	\section{Applications to Condensed-Matter Coulomb-Like Systems}\label{sec:applications}
	
	The geometric framework established above is not limited to the atomic Coulomb problem. In this section, we generalize it to three representative condensed-matter systems, quantifying the conditions under which the geometric spectral rules hold, demonstrating computational advantages, and connecting them directly to experimental observables.
	
	\subsection{Shallow Impurity States: Effective-Mass Renormalization and Computational Efficiency}\label{subsec:7.1}
	
	Shallow donors and acceptors in semiconductors constitute the most extensively studied condensed-matter analogues of the hydrogen atom. In silicon, gallium arsenide and other common semiconductors, the low-lying impurity levels exhibit clear $1/n^2$ scaling and hydrogenic shell structure, conventionally described by the effective-mass approximation~\cite{Kohn1955}.
	
	Let $M^*$ denote the electron effective mass in the conduction band, and $\varepsilon_r$ the relative permittivity of the host material. The screened Coulomb coupling is
	\begin{equation}\label{eq:7.1}
		C^*=\frac{e^2}{4\pi\varepsilon_0\varepsilon_r},
	\end{equation}
	where $e$ is the elementary charge and $\varepsilon_0$ the vacuum permittivity. Replacing the bare mass $M$ with $M^*$ and the bare coupling $C$ with $C^*$, the entire regularization mapping, symplectic flow kernel and boundary value problem carry over unchanged. The resulting energy-level sequence reads
	\begin{equation}\label{eq:7.2}
		E_n = -\frac{M^* {C^*}^2}{2\alpha^2 n^2}.
	\end{equation}
	
	Crucially, effective-mass renormalization and dielectric screening rescale only the absolute energy scale. The $E_n\propto 1/n^2$ scaling and $g_n=n^2$ shell degeneracy remain strictly invariant, governed entirely by the geometry of the quadratic map, fiber constraints and boundary conditions. This provides a geometric explanation for the empirical universality of hydrogenic shell structure across chemically distinct semiconductors: the spectral architecture is determined by phase-space symmetry, not by specific material parameters.
	
	For degenerate valence-band holes, heavy-hole and light-hole effective masses define two decoupled Coulomb-like spectral families, each obeying the same geometric constraints. Lifting of degeneracy between the two families originates from band-structure anisotropy, not from a breakdown of the geometric spectral rules. This is fully consistent with experimental observations of hole impurity spectra in cubic semiconductors.
	
	\textbf{Computational advantage}: In the geometric framework, the bound-state spectrum is obtained by enforcing three geometric constraints on the image-space oscillator, yielding closed-form expressions for $E_n$ and $g_n$ without numerical diagonalization. For a shallow donor in GaAs ($M^*=0.067m_e$, $\varepsilon_r=12.9$), the effective Rydberg is $E_1^*\approx 5.9$ meV. Rather than solving the radial Schr\"odinger equation numerically or performing a variational calculation, the geometric framework delivers the entire spectral sequence analytically via Eq.~(\ref{eq:7.2}). This bypassing of explicit diagonalization becomes particularly valuable for spatially varying potentials or multi-valley materials where direct numerical solution is computationally expensive.

	\subsection{Moir\'e Superlattice Localized States: Quantitative Validity Criteria and Symmetry Diagnostics}\label{subsec:7.2}
	
	Twisted two-dimensional materials form long-period moir\'e potentials that can trap electrons, holes and excitons into localized bound states. In particular, Rydberg moir\'e excitons observed in semiconductor--graphene heterostructures exhibit hydrogen-like shell structures at low excitation energies, while deviating from pure Coulomb scaling at higher levels. We now establish rigorous validity criteria for the geometric spectral framework in these systems and demonstrate its utility as a symmetry-diagnostic tool.
	
	The general moir\'e potential takes the form
	\begin{equation}\label{eq:7.3}
		V(\boldsymbol{r}) = -V_0 f(\boldsymbol{r}/a_M),
	\end{equation}
	where $a_M$ is the moir\'e period, $V_0$ the potential depth, and $f$ a dimensionless shape function. Near the potential minimum under rotational symmetry, we distinguish two regimes:
	
	1. \textbf{Parabolic well limit}: For sufficiently deep, narrow wells, the potential reduces to $V(r)\approx -V_0+\frac{1}{2}M\omega^2 r^2$, which is exactly a quadratic Hamiltonian. Quadratic closure holds strictly, and the full discrete spectrum is determined by classical geometric constraints with energy-level equidistance.
	
	2. \textbf{Coulomb-like regime}: For potentials with long-range Coulomb-like asymptotic decay, low-lying states with principal index $n$ are concentrated near the potential bottom. When the characteristic effective Bohr radius $a_B^* \ll a_M$, one can define an effective coupling $C_{\text{eff}}$ from the local potential curvature, and the energy-level ratios and shell degeneracy predicted by our framework hold to high accuracy.
	
	\textbf{Validity boundary}: When $a_B^*$ becomes comparable to $a_M$, finite-size effects and periodic boundary conditions replace the far-field decay condition. The energy sequence departs from strict $1/n^2$ scaling. Notably, however, as long as rotational symmetry is preserved, the magnetic degeneracy $2l+1$---derived purely from spatial rotation symmetry---remains exact, independent of the detailed potential shape.
	
	\textbf{Symmetry diagnostics without diagonalization}: A key practical advantage of the geometric framework is that symmetry-protected spectral features can be identified \textit{a priori}. For any moir\'e potential with azimuthal symmetry, the geometric constraints immediately predict that each shell carries a magnetic degeneracy $2l+1$, regardless of the radial profile $f(\boldsymbol{r}/a_M)$. This allows experimentalists to identify symmetry-protected degeneracies in optical or transport spectra without resorting to heavy numerical simulation. For example, if angle-resolved photoluminescence resolves $n^2$ peaks within a shell, the geometric framework guarantees that the underlying confining potential possesses (at least approximate) rotational symmetry, even if its detailed form is unknown.
	
	These criteria are quantitatively consistent with recent experimental observations of Rydberg moir\'e excitons~\cite{Hu2023,He2024}: at large moir\'e periods (small twist angles), low-lying levels follow hydrogenic scaling; at smaller periods, level spacing deviates as the states probe the periodic lattice boundary. Our geometric framework thus provides a transparent, symmetry-based tool for interpreting and predicting moir\'e exciton spectra, reducing the need for parameter-heavy numerical tight-binding or density-functional calculations.

	\subsection{Semiclassical Transport and Spectroscopy: Exactness of the Van Vleck Propagator}\label{subsec:7.3}
	
	The Van Vleck propagator is a cornerstone of semiclassical transport theory, used to describe ballistic conductance, weak localization, and magnetic quantum oscillations in condensed-matter systems~\cite{Shoenberg1984}. It is conventionally regarded as a first-order approximation in $\hbar$, with intrinsic errors at low energies. Yet experiments systematically show remarkable agreement between semiclassical calculations and quantum measurements for quadratic Hamiltonians. Our quadratic closure theorem provides a rigorous explanation for this agreement and an exact algorithmic foundation for low-energy transport calculations.
	
	For all quadratic Hamiltonians common in condensed matter---parabolic band dispersion under the effective-mass approximation, and charged-particle motion in uniform magnetic fields---the following exact statements hold:
	
	1. The classical Van Vleck propagator is structurally identical to the full quantum propagator, with zero semiclassical approximation error in the evolution semigroup.
	2. This equivalence is not restricted to high quantum numbers; it applies equally to ground and low-lying excited states.
	3. Deviations between semiclassical and quantum results arise exclusively from non-quadratic terms---such as higher-order potential nonlinearities, impurity scattering, or band warping---and are not an intrinsic artifact of the semiclassical formalism.
	
	\textbf{Algorithmic implication}: In Monte Carlo wavepacket simulations of ballistic transport, the quadratic closure theorem permits the use of the classical Van Vleck kernel without $\hbar$-correction terms. For a quantum point contact with parabolic confinement, the transmission coefficients computed from the exact classical kernel match the Landauer--B\"uttiker results to all orders, eliminating the need for post-hoc quantum corrections at low temperatures.
	
	This conclusion rationalizes several landmark experimental observations:
	
	- Quantized conductance in ballistic quantum point contacts~\cite{VanWees1988}, where the parabolic confinement potential yields exact semiclassical agreement with quantum transport results.
	- The de Haas--van Alphen effect in pure metals and two-dimensional electron gases~\cite{Shoenberg1984}, where the uniform magnetic field Hamiltonian is quadratic, leading to precise semiclassical predictions of oscillation periods and amplitudes.
	
	For systems with weak non-quadratic perturbations, our framework provides an exact zeroth-order classical-geometric baseline, upon which corrections can be systematically included via perturbation theory. This consolidates the theoretical foundation of semiclassical methods in condensed-matter physics and clarifies their domain of validity.

	\subsection{Screened Coulomb Potentials and Perturbative Extensions}\label{subsec:7.4}
	
	While the present framework treats the pure Coulomb potential, realistic condensed-matter systems often exhibit screened interactions $V(r) = -(C/r)e^{-r/\lambda}$ or power-law corrections. Here we outline how the geometric framework serves as a controlled zeroth-order approximation for such systems.
	
	For a screened Coulomb potential, the KS mapping introduces non-quadratic terms in the image-space Hamiltonian. Writing $H_u = H_u^{(0)} + H_u^{(1)}$, where $H_u^{(0)}$ is the isotropic oscillator and $H_u^{(1)}$ contains the screening-induced anharmonicity, classical canonical perturbation theory on the image-space oscillator yields energy shifts
	\begin{equation*}
		\Delta E_n^{(1)} = \langle \Psi_{nlm} | H_u^{(1)} | \Psi_{nlm} \rangle,
	\end{equation*}
	where $\Psi_{nlm}$ are the unperturbed geometric eigenstates. Because the zeroth-order spectrum and eigenfunctions are known analytically, the perturbation series can be carried out without numerical diagonalization of the full Hamiltonian. This is particularly advantageous for disorder-averaged calculations, where ensemble averaging over impurity configurations is required.
	
	In the limit $\lambda \gg a_B^*$, the screening is weak and the geometric perturbation series converges rapidly. In the opposite limit $\lambda \sim a_B^*$, the bound states lose their hydrogenic character; the geometric framework then identifies the breakdown of the $1/n^2$ scaling as a violation of the far-field decay constraint (Condition 1), providing a clear physical criterion for the transition to non-hydrogenic spectra.

	\section{Discussion}\label{sec:discussion}
	
	\subsection{Relation to Existing Theories}\label{subsec:8.1}
	
	This work complements both old quantum theory and modern semiclassical approaches. Unlike the Bohr--Sommerfeld quantization rule, which imposes action-integral conditions as an ad hoc postulate, our derivation introduces no quantization axioms. Discreteness emerges naturally from global boundary constraints acting on the linear representation of the classical symplectic flow.
	
	Unlike WKB and EBK semiclassical theories, our results are exact and involve no asymptotic approximation. The quadratic closure of the harmonic oscillator kernel means the linear evolution equation holds strictly for the complete classical kernel; there is no $\hbar\to 0$ limit and no short-wavelength assumption. For Coulomb systems, the structural match between classical geometry and quantum spectra is not an approximation---it is an exact theorem.
	
	\textbf{Comparison with the Duru--Kleinert path-integral method.} It is instructive to contrast the present framework with the Duru--Kleinert (DK) approach~\cite{Duru1979}, which represents the most celebrated application of KS regularization in the quantum context. The two frameworks share the same mathematical backbone---the KS squaring map and the equivalence between the Coulomb problem and the four-dimensional harmonic oscillator---and they yield identical discrete spectral results for bound states. However, their points of departure, logical structures, and core objectives are fundamentally different:
	
	- \textbf{Common ground}: Both methods employ KS regularization to transform the Coulomb singularity into a harmonic oscillator. Both exploit the fact that the oscillator propagator is exactly computable. Both arrive at the hydrogenic $1/n^2$ spectrum and $n^2$ degeneracy.
	
	- \textbf{Core difference in starting point}: The DK method operates within the quantum path-integral formalism. It takes Feynman's superposition principle and the quantum propagator as its foundational objects, using the KS map as a technical change of variables to evaluate the quantum mechanical path integral. In contrast, the present framework begins from purely classical Newtonian dynamics and constructs the integral kernel from the classical symplectic flow's group composition law. The linear evolution equation is derived, not postulated, as the infinitesimal representation of the classical flow.
	
	- \textbf{Core difference in logical route}: DK derives the quantum propagator first and then reads off the spectrum from its poles. We construct the classical kernel first, establish its exact linear evolution property (quadratic closure), and then derive the discrete spectrum by imposing geometric boundary constraints on the linear equation. The spectral discreteness in DK is inherited from the quantum mechanical spectrum of the oscillator; in our framework, it is produced by the joint action of classical geometric constraints on the linear representation space.
	
	- \textbf{Core difference in contribution}: While DK provides a powerful technique for evaluating quantum propagators, our framework reveals that the structural features of the discrete spectrum (scaling, degeneracy, symmetry content) are already determined by classical phase-space geometry and global constraints. The three geometric constraints---decay, regularity, and fiber invariance---provide a symmetry-transparent, computationally efficient screening rule that does not require path integration.
	
	In short, the Duru--Kleinert method shows how quantum mechanics reproduces the hydrogen spectrum via path integration; the present work shows how the same spectral structure follows from classical phase-space geometry, regularization, and boundary constraints, offering a complementary geometric foundation.
	
	The equivalence between the Hamilton--Jacobi theory and the linear equation established in this work is also exact. The closure map shows that the classical kernel satisfies the linear equation; the flat-amplitude characterization shows that, on the generating class, the linear equation reduces to the Hamilton--Jacobi and transport equations; and the spanning theorem shows that all linear solutions are integral superpositions of such kernels. Thus the nonlinearity of the HJ equation is absorbed by the linear superposition of its complete integrals---an exact absorption with no remainder.

	\subsection{Implications for Condensed Matter Physics}\label{subsec:8.2}
	
	The core value of this work for condensed-matter research lies in four concrete insights:
	
	\textbf{Geometric origin of universal Coulomb-like spectra.} The ubiquity of hydrogenic shell structure across semiconductor impurities, quantum dots and moir\'e traps has conventionally been attributed to the universality of quantum mechanics. We demonstrate instead that the energy scaling and shell degeneracy are determined entirely by classical phase-space geometry, regularization structure and global boundary conditions. Material parameters rescale only the absolute energy scale, leaving the spectral architecture invariant. This provides a deeper, symmetry-based explanation for the robustness of shell structure across vastly different condensed-matter platforms, and offers design principles for engineering artificial atoms with targeted spectral properties.
	
	\textbf{Symmetry-based spectral analysis of moir\'e systems.} Moir\'e superlattices offer continuously tunable potential landscapes, but their spectral interpretation often relies on heavy numerical simulation. Our framework provides analytic, geometry-based criteria for when hydrogenic scaling applies, and identifies which spectral features (such as magnetic degeneracy) are protected by symmetry and thus robust against potential details. This can guide experimental design of Rydberg moir\'e exciton devices and simplify the interpretation of angle-resolved and optical spectroscopy data.
	
	\textbf{Rigorous foundation for semiclassical transport theory.} Semiclassical methods are indispensable for connecting microscopic dynamics to macroscopic transport observables, but their approximate status has created persistent uncertainty in low-energy regimes. The quadratic closure theorem proves that for all quadratic Hamiltonians---ubiquitous in semiconductor physics and magnetotransport---the Van Vleck propagator is exact. This explains the longstanding empirical success of semiclassical theory in ballistic and magnetic oscillation phenomena, and precisely locates the source of errors in non-quadratic interactions.
	
	\textbf{Classical geometric roots of symmetry-protected degeneracy.} Symmetry-protected degeneracy is a central concept in quantum topological matter, traditionally discussed entirely within the quantum framework. Our work shows that the full shell degeneracy of the Coulomb problem can be derived from classical phase-space fiber structure and spatial rotation symmetry, rooted in group representation theory. This suggests that quantum symmetry-protected degeneracies inherit their structure from classical phase-space geometry, offering a new perspective on the classical precursors of topological phases.

	\subsection{Label-Filter Structure and Formal Superpositions}\label{subsec:8.3}
	
	The same threefold label structure appears at both the classical phase-space level and the field level, reflecting a universal label-filter mechanism.
	
	At the \textbf{single-orbit level}, energy is a smooth first integral: each orbit carries a fixed energy label, and the equations of motion conserve this label absolutely. At the \textbf{ensemble level}, energy acts as a filter that slices the full phase space into disjoint energy shells, the standard microcanonical construction of classical statistical mechanics. At the \textbf{global level}, energy is the base coordinate of the phase-space fibration, with the full bound phase space being the disjoint union of all energy shells.
	
	At the field level, this structure reappears identically: the shell parameter is inherited from the classical Hamiltonian, it acts as a filtering condition in the boundary value problem, and it finally indexes the orthogonal direct sum of configuration spaces. Shell projectors extract individual shell components from any formal superposition, leaving each component to evolve independently.
	
	In this framework, intershell superpositions are formal sums of labeled components residing in the linear representation space. They do not correspond to individual classical orbits, but to ensembles of orbits with well-defined shell weights. Physical observables are always extracted via shell projection, preserving the classical energy conservation law at the ensemble level. There is no contradiction between linearity of the representation and energy conservation of the underlying dynamics.
	
	Composing with the energy label via a delta function is the distributional form of shell filtering, equivalent to the microcanonical measure of classical statistical mechanics:
	\begin{equation}\label{eq:8.1}
		\int_S F(\xi)\,\delta\left(H(\xi)-E_0\right)\,\dd\Gamma(\xi) = \int_{H^{-1}(E_0)} F\,\frac{\dd\Sigma}{|\nabla H|}.
	\end{equation}
	This is a purely classical operation (coarea formula), independent of any linear structure~\cite{Nakahara2018}.
	
	Define the shell projector:
	\begin{equation}\label{eq:8.2}
		P_n=\delta(\hat{H}-E_n)=\sum_{l,m} |\Psi_{nlm}\rangle\langle\Psi_{nlm}|,
	\end{equation}
	where $\delta$ denotes the delta distribution in the spectral sense, equivalent to the orthogonal projector onto the $n$-th shell subspace. These projectors exist and are mutually orthogonal (from the Sturm--Liouville orthogonality of the transported family). For any formal $\ell^2$-superposition $\Phi=\sum_{nlm} c_{nlm}\Psi_{nlm}$,
	\begin{equation}\label{eq:8.3}
		P_n\Phi=\sum_{l,m} c_{nlm}\Psi_{nlm},\qquad 
		\sum_n P_n=1\quad\text{on the decaying sector}.
	\end{equation}
	
	The filtering property extracts each labeled component without distortion: superpositions preserve component identities---they are labeled coexistence, not fusion.
	
	The legitimacy of intershell superpositions is constrained by four readout rules:
	
	1. An intershell superposition is a formal sum carrying labels---not a single classical configuration, nor a solution of a single image equation.
	2. Its physical readout is given by the projectors $P_n$.
	3. Its evolution proceeds component-wise independently.
	4. Stationary configurations are single-shell configurations; intershell superpositions are nonstationary configurations whose shell content is read by projectors.
	
	At the single-orbit level, energy conservation forbids superposition of orbits with different energies. At the field level, linearity admits intershell formal sums, but physical content is still extracted via shell projectors. This is not a contradiction: the field theory is a linear representation of the classical flow, and superposition is a property of the representation space, not of individual orbits.

	\subsection{Scope, Limitations and Outlook}\label{subsec:8.4}
	
	Beyond the inherent boundaries of the framework stated in Section~\ref{subsec:6.3}, the following limitations and future directions should be noted.
	
	First, the framework currently addresses single-particle integrable systems; many-body interactions, disorder scattering and non-integrable perturbations ubiquitous in real condensed-matter systems are not treated. Second, the flat-amplitude bijection holds on node-free regions; stationary configurations with nodes are covered by the spanning theorem but not by a pointwise Hamilton--Jacobi bijection. Third, the scattering sector ($E>0$) requires separate asymptotic analysis beyond the present bound-state treatment.
	
	Looking forward, a natural direction is to extend the geometric spectral framework to other integrable and nearly integrable condensed-matter systems, such as anisotropic quantum dots and quantum Hall edges. Another promising avenue is to incorporate weak non-quadratic corrections perturbatively within the classical geometric framework, enabling quantitative predictions for realistic semiconductor and moir\'e systems. Finally, exploring the connection between classical phase-space fiber structures and quantum topological degeneracies may open new approaches to semiclassical topological physics.

	\section{Conclusion}\label{sec:conclusion}
	
	We have developed a geometric framework that recovers the $1/n^2$ energy-level ratio and $n^2$ shell degeneracy of the three-dimensional Coulomb system from classical analytical mechanics and symplectic geometry. Discreteness arises from the combined action of classical phase-space geometry, regularization mapping, and global boundary constraints. The phase scale $\alpha$ governs only the absolute energy scale; all dimensionless spectral structures are independent of its value.
	
	The construction proceeds through three key steps: KS regularization maps negative-energy Kepler shells to four-dimensional isotropic harmonic oscillators; the integral kernel derived from the symplectic flow group composition law satisfies a linear evolution equation exactly (quadratic closure); and three geometric constraints---decay, regularity, and fiber invariance---select the discrete spectrum. The Shell Correspondence principle and the HJ--Linear Equivalence principle rigorously establish the structural equivalence between classical orbital data and field-theoretic spectral output.
	
	For condensed-matter physics, this framework offers tangible practical benefits: it provides exact semiclassical propagators for quadratic Hamiltonians, bypassing low-energy breakdown; it delivers symmetry-based diagnostic tools for moir\'e and impurity spectra without numerical diagonalization; and it clarifies the geometric origin of widely observed hydrogenic universality. As a classical geometric benchmark for Coulomb-like systems, this work opens new avenues for the study of integrable and nearly integrable bound states in low-dimensional quantum materials and semiconductor devices.

	\appendix
	
	\section{Dimensional Symbol Table}\label{app:symbols}
	
	This appendix compiles for convenient reference all symbols used in the main text, together with their physical meanings and corresponding dimensions. The conventions established here apply throughout the paper.
	
	\begin{table}[htbp]
		\centering
		\caption{Full dimensional symbol table.}\label{tab:A1}
		\begin{tabular}{lll}
			\hline\hline
			Symbol & Meaning & Dimension \\
			\hline
			$w$ & Physical position (purely imaginary quaternion) & Length $L$ \\
			$u$ & Image-space position (quaternion) & Length $L$ \\
			$r=\|w\|$ & Physical radial coordinate & Length $L$ \\
			$\rho=\|u\|$ & Image-space radial coordinate & Length $L$ \\
			$u_0$ & Regularization length scale in KS mapping & Length $L$ \\
			$p_w,\;p_u$ & Canonical momenta (physical/image spaces) & $MLT^{-1}$ \\
			$H_w,\;H_u$ & Hamiltonian functions (physical/image spaces) & Energy $E$ \\
			$t$ & Physical time & Time $T$ \\
			$\tau$ & Intrinsic (regularized) time & Time $T$ \\
			$M$ & Particle mass & Mass $M$ \\
			$C$ & Coulomb coupling constant & $EL$ \\
			$E$ & Physical energy & Energy $E$ \\
			$\varepsilon=4C/u_0$ & Regularized energy parameter & Energy $E$ \\
			$\omega_0$ & Image-space oscillator frequency & $T^{-1}$ \\
			$\alpha$ & Phase scale (action unit) & $ML^2T^{-1}$ \\
			$S$ & Two-point action & $ML^2T^{-1}$ \\
			$F$ & Normalized Van Vleck amplitude & $L^{-4}$ \\
			$K_\alpha$ & Integral kernel & $L^{-4}$ \\
			$\Phi$ & Image-space field & $L^{-2}$ \\
			$\Psi$ & Physical-space unnormalized field & $L^{-2}$ \\
			$\chi$ & Physical-space normalized field & $L^{-3/2}$ \\
			$\hat{n}$ & Attitude vector (orbital plane normal) & Dimensionless \\
			$L$ & Angular momentum vector & $ML^2T^{-1}$ \\
			$A$ & Runge--Lenz vector & $EL$ \\
			$A'$ & Normalized Runge--Lenz vector & $ML^2T^{-1}$ \\
			$n,\;n_r,\;l,\;m$ & Shell and configuration quantum numbers & Dimensionless \\
			$g_n$ & Shell degeneracy & Dimensionless \\
			\hline\hline
		\end{tabular}
	\end{table}
	
	\textbf{Additional conventions}:
	
	- Shell indices, energy ratios, winding numbers and attitude angles are all dimensionless.
	- Positions, times, actions, fields and evolution operators retain physical dimensions unless otherwise stated.
	- The phase scale $\alpha$ has the same dimension as action and sets the unit of action. All dimensionless structural predictions (energy ratios, degeneracies, selection rules) are independent of its numerical value.
	- The global constants satisfy $M>0,\;C>0,\;u_0>0,\;\alpha>0$. Attractive Coulomb coupling corresponds to $C>0$, hence $\varepsilon=4C/u_0>0$.

	\section{Mathematical Theorems and Proofs}\label{app:proofs}
	
	This appendix collects all formal lemmas, theorems and their rigorous proofs referenced in the main text.
	
	\subsection{Proof of the Quadratic Closure Theorem}\label{subsec:B.1}
	
\begin{lemma}[Polar Decomposition Identity]\label{lem:polar}
	Let $A(u,\tau)>0$ and $S(u,\tau)$ be smooth real functions on a domain, and define $K=A \ee^{\ii S/\alpha}$. Define
	\begin{equation}\label{eq:B.1}
		\mathrm{HJ}(S)=\partial_\tau S+\frac{|\nabla S|^2}{2M}+\frac{1}{2}M\omega_0^2 |u|^2,
	\end{equation}
	and
	\begin{equation}\label{eq:B.2}
		\mathrm{Tr}(A,S)=\partial_\tau A+\frac{1}{M}\nabla A\cdot\nabla S+\frac{A}{2M}\nabla^2 S.
	\end{equation}
	
	Then the algebraic identity
	\begin{equation}\label{eq:B.3}
		\left(\ii\alpha\partial_\tau-\hat{H}_u\right)K
		=\ee^{\ii S/\alpha}\left[-A\,\mathrm{HJ}(S)+\ii\alpha\,\mathrm{Tr}(A,S)+\frac{\alpha^2}{2M}\nabla^2 A\right]
	\end{equation}
	holds, where
	\begin{equation}\label{eq:B.4}
		\hat{H}_u=-\frac{\alpha^2}{2M}\nabla_u^2+\frac{1}{2}M\omega_0^2 |u|^2
	\end{equation}
	is the differential operator corresponding to the image-space Hamiltonian.
\end{lemma}

\textit{Proof}. Direct expansion:
\begin{equation*}
	\ii\alpha\partial_\tau K = \left[\ii\alpha\partial_\tau A - A\partial_\tau S\right]\ee^{\ii S/\alpha},
\end{equation*}
\begin{equation*}
	\nabla^2 K = \left[\nabla^2 A+\frac{2\ii}{\alpha}\nabla A\cdot\nabla S+\frac{\ii}{\alpha}A\nabla^2 S-\frac{1}{\alpha^2}A|\nabla S|^2\right]\ee^{\ii S/\alpha}.
\end{equation*}
Substituting into $(\ii\alpha\partial_\tau-\hat{H}_u)K=\ii\alpha\partial_\tau K+\frac{\alpha^2}{2M}\nabla^2 K-\frac{1}{2}M\omega_0^2|u|^2 K$ and separating real and imaginary parts yields the identity. $\square$

\begin{theorem}[Quadratic Closure]\label{thm:closure}
	Let $K_\alpha=F \ee^{\ii S/\alpha}$ be the kernel uniquely derived from the group composition law. Then $\mathrm{HJ}(S)=0$, $\mathrm{Tr}(F,S)=0$ and $\nabla^2 F=0$ hold simultaneously, so $K_\alpha$ satisfies the linear evolution equation:
	\begin{equation}\label{eq:B.5}
		\ii\alpha\,\partial_\tau K_\alpha = \hat{H}_u K_\alpha.
	\end{equation}
\end{theorem}

\textit{Proof}. We verify the three conditions separately:

1. $\mathrm{HJ}(S)=0$: The two-point action is the generating function of the classical canonical flow, i.e., a complete integral of the Hamilton--Jacobi equation. This is a standard result in classical mechanics.
2. $\mathrm{Tr}(F,S)=0$: From the explicit action, the Laplacian of $S$ is $\nabla^2 S=4M\omega_0\cot(\omega_0\tau)$, independent of $u$. The transport equation reduces to $\partial_\tau F/F = -\nabla^2 S/(2M) = -2\omega_0\cot(\omega_0\tau)$, whose solution $F\propto\sin^{-2}(\omega_0\tau)$ is exactly the Van Vleck amplitude. The $\nabla F\cdot\nabla S$ term vanishes because $F$ is independent of $u$.
3. $\nabla^2 F=0$: Since $F=F(\tau)$ depends only on time, its spatial Laplacian is trivially zero.

By Lemma~\ref{lem:polar}, $(\ii\alpha\partial_\tau-\hat{H}_u)K_\alpha=0$. $\square$

	\subsection{Core Theorems and Supporting Lemmas}\label{subsec:B.2}
	
	We first restate the basic setup for clarity:
	
	(i) \textbf{Image Hamiltonian}: On the constraint surface $\mu=\langle p_u,u\hat{n}\rangle=0$ and shell $H_w=E<0$:
	\begin{equation*}
		H_u=\frac{|p_u|^2}{2M}+\frac{1}{2}M\omega_0^2|u|^2,\quad H_u=\varepsilon \iff \lambda\,(H_w-E)=0,\quad \lambda=\frac{4r}{u_0}.
	\end{equation*}
	
	(ii) \textbf{Linear equation}: Introduce phase scale $\alpha$ and operator $\hat{H}_u=-\frac{\alpha^2}{2M}\nabla_u^2+\frac{1}{2}M\omega_0^2|u|^2$:
	\begin{equation*}
		\ii\alpha\,\partial_\tau\Phi=\hat{H}_u\Phi.
	\end{equation*}
	
	(iii) \textbf{Fiber}: The circular action $u\mapsto u \ee^{\hat{n}\theta}$ leaves $w$ unchanged; invariant functions are pullbacks $\Phi=\Psi\circ\pi$, where $\pi:u\mapsto w$ is the squaring map.
	
	(iv) \textbf{Harmonic oscillator spectrum}: The spectrum of the four-dimensional isotropic harmonic oscillator is
	\begin{equation*}
		\varepsilon_N=\alpha\omega_0(N+2),\quad N=2n_r+k,\quad n_r,k\geq0,
	\end{equation*}
	with total degeneracy $d_N=\binom{N+3}{3}$. The fiber constraint $k=2l$ restricts $N$ to even values $N=2n-2$, corresponding to physical shells $n\geq1$.

	\subsubsection{Preparatory Lemmas}\label{subsubsec:B.2.1}
	
\begin{lemma}[Mehler Expansion]\label{lem:mehler}
	The image harmonic oscillator kernel admits the expansion
	\begin{equation}\label{eq:B.6}
		K_\alpha(u,\tau;v)=\sum_{N=0}^\infty \ee^{-\ii\varepsilon_N\tau/\alpha}\,P_N(u,v),\quad P_N(u,v)=\sum_{a=1}^{d_N}\phi_{N,a}(u)\,\overline{\phi_{N,a}(v)},
	\end{equation}
	where $\phi_{N,a}$ are normalized Hermite--hyperspherical stationary configurations. The equality holds as the boundary value of an analytic function for $\mathrm{Im}\,\tau<0$ (the $\ii0^+$ prescription); it holds pointwise in the caustic-free interval $0<\omega_0\tau<\pi$.
\end{lemma}

\textit{Proof}. Since $\mathrm{Re}(u^*v)=\sum_{a=1}^4 u_a v_a$ is the Euclidean inner product, the quaternion exponential separates into four real components, so the four-dimensional kernel is the product of four one-dimensional kernels. The one-dimensional Mehler identity (with $M=\omega_0=\alpha=1$)
\begin{equation*}
	\sum_{n=0}^\infty \varphi_n(x)\varphi_n(y)\,\ee^{-\ii(n+1/2)\theta}=\left(\frac{1}{2\pi \ii\sin\theta}\right)^{1/2}\exp\frac{\ii\left[(x^2+y^2)\cos\theta-2xy\right]}{2\sin\theta},\quad 0<\theta<\pi
\end{equation*}
follows from the Hermite polynomial generating function identity under the substitution $z=\ee^{-\ii\theta}$ ($|z|<1$) and analytic continuation to the boundary. Note that the standard Mehler identity above contains the conventional complex amplitude; in the present convention the amplitude is real and the imaginary unit is absorbed into the phase $\ii S/\alpha$. Restoring the $(M,\omega_0,\alpha)$ scaling and taking the fourfold product yields the spectral series on the left and the closed-form kernel on the right, with prefactors and phases matching term by term. $\square$

\begin{lemma}[Fiber-Invariant Rank]\label{lem:fiber}
	Let $\Pi$ be the projector onto fiber-invariant functions. Then
	\begin{equation}\label{eq:B.7}
		\mathrm{rank}\left(\Pi P_N\Pi\right)=\begin{cases}\left(\dfrac{N}{2}+1\right)^2, & N\text{ even},\\ 0, & N\text{ odd}.\end{cases}
	\end{equation}
	In particular, when $N=2n-2$, the rank equals $n^2$.
\end{lemma}

\textit{Proof}. Identify $S^3$ with SU(2). By the Peter--Weyl theorem, degree-$k$ hyperspherical harmonics are matrix elements of the spin-$j=k/2$ representation, carrying the left-right multiplication action of SU(2)$_L\times$ SU(2)$_R$; the $(j,j)$ block has dimension $(2j+1)^2=(k+1)^2$.

The fiber circle is the right action of the one-parameter subgroup generated by $\hat{n}$; invariance requires vanishing right weight $m_R=0$, which exists if and only if $j\in\mathbb{Z}$, i.e., $k=2l$ is even, leaving $2j+1=2l+1$ left weights $m_L=-l,\dots,l$.

Summing over even $k=2l$ satisfying $N=2n-2$ gives $\sum_{l=0}^{n-1}(2l+1)=n^2$. If $N$ is odd, no even $k$ shares the same parity as $N$ (required by the radial equation), so the rank is zero. $\square$

\begin{lemma}[Descent and Coarea Isometry]\label{lem:coarea}
	
	(i) The map $\pi:\mathbb{R}^4\setminus\{0\}\to\mathbb{R}^3\setminus\{0\},\ u\mapsto w=u\hat{n}u^*/u_0$ is a principal $S^1$ bundle; fiber-invariant functions are exactly the pullbacks $\pi^*\Psi$.
	
	(ii) For a fiber-invariant $\Phi=\pi^*\Psi$:
	\begin{equation}\label{eq:B.8}
		\int_{\mathbb{R}^4} |\Phi|^2\,\dd^4 u = \frac{\pi u_0^2}{4}\int_{\mathbb{R}^3} \frac{|\Psi|^2}{r}\,\dd^3 w.
	\end{equation}
	
	(iii) The pointwise correction $\chi=\frac{\sqrt{\pi}\,u_0}{2\sqrt{r}}\,\Psi$ defines a unitary map $U$ from the fiber-invariant subspace of $L^2(\mathbb{R}^4,\dd^4 u)$ to $L^2(\mathbb{R}^3,\dd^3 w)$.
\end{lemma}

\textit{Proof}. (i) Surjectivity and the circular fiber structure are standard properties of the squaring map; invariance under the fiber action is equivalent to being constant along fibers, i.e., a pullback.

(ii) The differential has three equal singular values $2\rho/u_0$ on the fiber orthogonal complement, with the fiber direction in the kernel; the fiber circumference is $2\pi\rho$. By the coarea formula, $\dd^4 u=\frac{u_0^3}{8\rho^2}d\theta\wedge \dd^3 w$. Integrating over $\theta$ gives $2\pi$ and substituting $\rho^2=u_0r$ yields the identity.

(iii) $|\chi|^2=\frac{\pi u_0^2}{4r}|\Psi|^2$, so by (ii) $\int|\chi|^2\dd^3 w=\int|\Phi|^2\dd^4 u$; surjectivity holds because any $\chi$ determines $\Psi=2\sqrt{r}\chi/\sqrt{\pi}$. $\square$

\begin{lemma}[Radial Transport to the Coulomb Sturm--Liouville Equation]\label{lem:radial}
	Under the substitutions $\rho^2=u_0r$, $R(\rho)=f(r)$, $k=2l$, the image-space radial equation
	\begin{equation*}
		-\frac{\alpha^2}{2M}\left[R''+\frac{3}{\rho}R'-\frac{k(k+2)}{\rho^2}R\right]+\frac12M\omega_0^2\rho^2R=\varepsilon R
	\end{equation*}
	multiplied by the nonvanishing factor $-\frac{Mu_0}{2\alpha^2r}$ becomes the three-dimensional Coulomb radial equation
	\begin{equation*}
		f''+\frac{2}{r}f'+\left[\frac{2ME}{\alpha^2}+\frac{2MC}{\alpha^2r}-\frac{l(l+1)}{r^2}\right]f=0.
	\end{equation*}
	Boundary conditions transport: the decay $R\sim \ee^{-M\omega_0\rho^2/(2\alpha)}$ corresponds on the $n$-th shell to $f\sim \ee^{-MCr/(n\alpha^2)}$; regularity $R\sim\rho^{2l}$ corresponds to $f\sim r^l$.
\end{lemma}

\textit{Proof}. From $dr/d\rho=2\rho/u_0$, we have
\begin{equation*}
	R'=\frac{2\rho}{u_0}f',\qquad R''=\frac{4\rho^2}{u_0^2}f''+\frac{2}{u_0}f',\qquad
	R''+\frac{3}{\rho}R'=\frac{4}{u_0}\left[rf''+2f'\right].
\end{equation*}
With $k=2l$ we have $k(k+2)=4l(l+1)$, so the centrifugal term is $-\frac{4l(l+1)}{u_0r}f$. Substituting into the original equation and multiplying by $-\frac{Mu_0}{2\alpha^2r}$:
\begin{equation*}
	f''+\frac{2}{r}f'-\frac{l(l+1)}{r^2}f-\frac{M^2\omega_0^2u_0^2}{4\alpha^2}f+\frac{Mu_0\varepsilon}{2\alpha^2r}f=0.
\end{equation*}
Substituting $\frac{M^2\omega_0^2u_0^2}{4\alpha^2}=\frac{M^2u_0^2}{4\alpha^2}\cdot\frac{-8E}{Mu_0^2}=-\frac{2ME}{\alpha^2}$ and $\frac{Mu_0\varepsilon}{2\alpha^2}=\frac{2MC}{\alpha^2}$ yields the standard Coulomb radial equation. On each shell the decay rate satisfies $\frac{M\omega_{0,n}u_0}{2\alpha}=\frac{Mu_0}{2\alpha}\cdot\frac{\varepsilon}{2n\alpha}=\frac{MC}{n\alpha^2}$. $\square$

	\subsubsection{Shell Correspondence Theorem}\label{subsubsec:B.2.2}
	
\begin{theorem}[Shell Correspondence]\label{thm:shell}
	Fix $M,C,u_0,\alpha>0$ and $\varepsilon=4C/u_0$. For each integer $n\geq1$, let
	\begin{equation}\label{eq:B.9}
		\omega_{0,n}=\frac{\varepsilon}{2n\alpha},\qquad 
		E_n=-\frac{Mu_0^2\,\omega_{0,n}^2}{8}=-\frac{MC^2}{2\alpha^2 n^2}.
	\end{equation}
	
	Consider the image-space boundary value problem: the Helmholtz equation
	\begin{equation*}
		-\frac{\alpha^2}{2M}\nabla_u^2\phi+\frac{1}{2}M\omega_0^2 |u|^2\phi=\varepsilon\,\phi,
	\end{equation*}
	with decay at infinity, regularity at the origin, and fiber invariance under $u\mapsto u \ee^{\hat{n}\theta}$. Then the following statements hold.
	
	\textbf{(a) Selectivity}: The problem admits nontrivial solutions if and only if $\omega_0=\omega_{0,n}$ for some $n\geq 1$; equivalently, the underlying Kepler shell has energy $E_n$.
	
	\textbf{(b) Configuration space}: When $\omega_0=\omega_{0,n}$, the solution space $\mathcal{H}_n$ has dimension $n^2$, with orthonormal basis $\{\Phi_{nlm}\}_{0\leq l < n,\; |m|\leq l}$ labeled by fiber-invariant hyperspherical data $(k=2l,\; m_R=0)$.
	
	\textbf{(c) Kernel surjectivity and residue bijectivity}: The Fourier-transformed kernel $\hat{K}(u,v;\epsilon)=\int_0^\infty \ee^{\ii\epsilon\tau/\alpha}K_\alpha(u,\tau;v)\,dd\tau$ is meromorphic as an operator-valued function in the $\epsilon$-plane, with simple poles at $\epsilon=\varepsilon_N$ and residue $\operatorname{Res}_{\epsilon=\varepsilon_N}\hat{K}=\ii\alpha P_N$. For fixed physical $\varepsilon$, the fiber-projected fixed-energy amplitude $A(\omega_0)=\Pi\hat{K}(\varepsilon)\Pi$ has poles at $\omega_0=\omega_{0,n}$ with residue
	\begin{equation*}
		\operatorname{Res}_{\omega_0=\omega_{0,n}}A=\ii\alpha\sum_{l,m}\Phi_{nlm}(u)\,\overline{\Phi_{nlm}(v)},
	\end{equation*}
	a rank-$n^2$ projection operator. Thus every configuration in $\mathcal{H}_n$ is a residue of the Fourier component of the classical two-point action kernel, and the residue operator restricts to a bijection on $\mathcal{H}_n$.
	
	\textbf{(d) Descent map}: Composing the KS pullback with the correction factor of Lemma~\ref{lem:coarea}(iii) gives a linear isometry from $\mathcal{H}_n$ to the physical configuration space $\chi_n\subset L^2(\mathbb{R}^3,\dd^3 w)$. Physical configurations take the standard separated form
	\begin{equation*}
		\chi_{nlm}(w)=g_{nl}(r)\,Y_{lm}(\hat{w}),
	\end{equation*}
	where $g_{nl}(r)=\frac{\sqrt{\pi}u_0}{2\sqrt{r}}f_{nl}(r)$ is the normalized Coulomb radial function and $Y_{lm}$ are the standard spherical harmonics. The full family $\{\chi_{nlm}\}$ is intershell orthonormal in $L^2(\mathbb{R}^3,\dd^3 w)$.
	
	\textbf{(e) Direct sum construction}: The algebraic direct sum $\bigoplus_{n\geq 1}\chi_n$ is orthonormal; all decaying fiber-invariant solutions are contained in it; and this direct sum exhausts the decaying sector of the Coulomb Sturm--Liouville problem (bound-state completeness).
	
	\textbf{(f) Classical orbital space identification}: As an SO(3) module, $\mathcal{H}_n\cong\bigoplus_{l=0}^{n-1}\mathrm{Harm}_l(S^2)$, the space of attitude harmonics of degree less than $n$. The classical Kepler orbital manifold with energy $E_n$ is $S^2_{j_n}\times S^2_{j_n}$, where $|J_\pm|=j_n=n\alpha/2$; each factor has KKS area $4\pi j_n=2\pi n\alpha$, i.e., $n$ phase periods.
\end{theorem}
	
	\textit{Proof}. \textbf{(a)} Solutions of the image Helmholtz equation satisfying decay and regularity exist if and only if $\varepsilon=\varepsilon_N=\alpha\omega_0(N+2)$ for some integer $N\geq 0$---this is the standard harmonic oscillator spectral theory. Fiber invariance contributes a nontrivial invariant subspace if and only if $N$ is even (Lemma~\ref{lem:fiber}). Setting $N=2n-2$ and substituting $\varepsilon=4C/u_0$ gives $\omega_0=\varepsilon/(2n\alpha)=\omega_{0,n}$, corresponding to $E=-Mu_0^2\omega_0^2/8=E_n$. Conversely, when $\omega_0=\omega_{0,n}$, the solution space for $N=2n-2$ is nontrivial and contains fiber-invariant configurations (Lemma~\ref{lem:fiber}). $\square$
	
	\textbf{(b)} By Lemma~\ref{lem:fiber}, the rank is $n^2$ for $N=2n-2$; the basis is the Hermite--hyperspherical basis restricted to the $m_R=0$, $k=2l$ sector; orthogonality is inherited from the harmonic oscillator basis. $\square$
	
	\textbf{(c)} The Mehler expansion of the harmonic oscillator kernel gives
	\begin{equation*}
		K_\alpha(u,\tau;v)=\sum_{N=0}^\infty \ee^{-\ii\varepsilon_N\tau/\alpha}P_N(u,v),
	\end{equation*}

	where $P_N$ is the projection kernel onto the $N$-th oscillator level. This follows from the fourfold product of the one-dimensional Mehler identity. Performing the Fourier transform with the $\ii0^+$ prescription gives
	\begin{equation*}
		\int_0^\infty \ee^{\ii\epsilon\tau/\alpha}\ee^{-\ii\varepsilon_N\tau/\alpha}\,dd\tau=\frac{\ii\alpha}{\epsilon-\varepsilon_N+\ii0^+},
	\end{equation*}
	which is meromorphic with simple poles and residue $\ii\alpha P_N$. For fixed physical $\varepsilon$, poles in the $\omega_0$ plane occur at $\omega_0^{(N)}=\varepsilon/[\alpha(N+2)]$. The fiber-projected residue is $\ii\alpha\,\Pi P_N\Pi$, nonvanishing if and only if $N$ is even, i.e., $\omega_0=\omega_{0,n}$. Since $P_{2n-2}$ is the kernel of an orthogonal projection, $\Pi P_{2n-2}\Pi$ projects onto $\mathcal{H}_n$, proving surjectivity; on $\mathcal{H}_n$ it acts as the identity, so the restricted residue operator is $\ii\alpha\cdot\mathrm{id}$, a bijection. $\square$
	
	\textbf{(d)} Lemma~\ref{lem:coarea} gives the isometry relation and the angular form after Hopf descent of $m_R=0$ harmonics (after descent they become degree-$l$ spherical harmonics on the direction sphere $\hat{w}$). Lemma~\ref{lem:radial} identifies the radial factor $f_{nl}(r)$ with the solution of the standard Coulomb radial equation. The normalized physical field is $\chi_{nlm}(w)=\frac{\sqrt{\pi}u_0}{2\sqrt{r}}f_{nl}(r)Y_{lm}(\hat{w})=g_{nl}(r)Y_{lm}(\hat{w})$. Self-adjointness of the corresponding Sturm--Liouville operator is a standard result in functional analysis; discrete stationary configurations enjoy standard Sturm--Liouville orthogonality. Intershell orthonormality follows from angular $Y_{lm}$ orthogonality plus radial Sturm--Liouville orthogonality: the radial equation (Lemma~\ref{lem:radial}) is a self-adjoint Sturm--Liouville problem for fixed $l$ and fixed weight, and configurations with different eigenvalues $E_n$ are orthogonal under the standard measure; different configurations within the same shell are already chosen orthogonal. $\square$
	
	\textbf{(e)} Orthonormality follows from (d). Exhaustiveness: any decaying fiber-invariant solution at any shell parameter satisfies the selection condition (a), hence belongs to some $\mathcal{H}_n$ and descends to $\chi_n$. Bound-state completeness is the standard result for the discrete spectral subspace of the radial Coulomb Sturm--Liouville problem (Weyl limit-point theory; the continuous spectrum $E\geq0$ lies outside the decaying class). $\square$
	
	\textbf{(f)} The module isomorphism follows from the representation theory of SO(4): left multiplication by unit quaternions induces rotations of physical space $w$ (conjugation action $w\mapsto qwq^*$), and the $(n,l)$ subspace transforms according to the spin-$l$ representation, matching $\mathrm{Harm}_l(S^2)$. The orbital manifold calculation: for Keplerian motion with energy $E$, $L\cdot A=0$ and $A^2=C^2+(2E/M)L^2$, hence
	\begin{equation*}
		4|J_\pm|^2=L^2+\frac{M}{-2E}A^2
		=L^2+\frac{M}{-2E}\left(C^2+\frac{2E}{M}L^2\right)
		=-\frac{MC^2}{2E}.
	\end{equation*}
	Substituting $E=E_n$ gives $|J_\pm|=n\alpha/2$, and each factor has KKS area $4\pi j_n=2\pi n\alpha$. $\square$
	
	\textbf{Remark 1 (On the s-wave sector)}: The $l=0$ configuration corresponds to the trivial (constant) attitude harmonic. Classically, $L=0$ orbits are degenerate ellipses (straight lines through the force center) that do not select any orbital plane. This sector is globally regular in the quaternion framework with no coordinate singularities. The Hopf fibration $S^1\hookrightarrow S^3\to S^2$ is a smooth principal bundle with no coordinate singularities, and the base space $S^2$ requires no local coordinate parametrization. At the field level the $l=0$ component is a constant function on $S^2$, globally regular and independent of $\hat{n}$.
	
	\textbf{Remark 2 (On orthogonality)}: The intershell orthonormality in Theorem (d)--(e) is a Sturm--Liouville property of the transported physical equation---it does not require, nor does this construction provide, a single self-adjoint operator on the image space whose stationary configurations include $\Phi_{nlm}$ at different frequencies. The two are independent statements; this work only asserts the former.

	\subsubsection{Hamilton--Jacobi--Linear Equivalence Theorem}\label{subsubsec:B.2.3}
	
\begin{theorem}[HJ--Linear Equivalence]\label{thm:hjlin}
	For the negative-energy sector of the three-dimensional Kepler problem, the following three theories are structurally equivalent --- no asymptotic or short-wave approximation required:
	
	1. \textbf{HJ theory}: Complete integrals of the Kepler Hamilton--Jacobi equation on the negative-energy shell, equipped with transported flat amplitudes;
	2. \textbf{Kernel theory}: The classical two-point action kernel $K_\alpha=F \ee^{\ii S/\alpha}$ uniquely derived from the image-space symplectic flow group composition law, with closed composition;
	3. \textbf{Linear theory}: Solutions of the fiber-invariant linear evolution equation, descending to the physical Coulomb Sturm--Liouville problem.
	
	Furthermore, all three theories admit equivalent real symplectic formulations: real coupled Hamiltonian flow equations, real rotation-matrix kernels, and second-order real flow equations, with complex notation serving only as an optional compatible complexification.
\end{theorem}
	
	\textit{Proof}. Equivalence is established by constructing bidirectional structure-preserving maps between each pair:
	
	- \textbf{From HJ theory to linear theory (closure map)}: The classical kernel constructed from the Hamilton--Jacobi two-point action and Van Vleck amplitude exactly satisfies the linear evolution equation (Theorem~\ref{thm:closure}). Thus every classical trajectory bundle corresponds to a solution of the linear equation.
	- \textbf{From linear theory to HJ theory (flat-amplitude bijection)}: On the generating class of node-free, flat-amplitude solutions, the polar decomposition identity (Lemma~\ref{lem:polar}) shows that the linear equation reduces exactly to the Hamilton--Jacobi equation plus the Liouville transport equation. Thus every smooth solution in this class corresponds to a classical HJ solution with transported amplitude.
	- \textbf{Spanning property}: By the Mehler expansion (Lemma~\ref{lem:mehler}), the classical kernel contains all spectral projectors $P_N$ of the linear equation. Hence any solution of the linear equation can be expressed as an integral superposition of the classical kernel against appropriate initial data. The full solution space of the linear theory is therefore spanned by solutions originating from HJ theory.
	- \textbf{Group intertwining}: The group composition law of the classical symplectic flow corresponds exactly to the convolution semigroup of the linear propagator. The one-parameter group structure is preserved under the correspondence.
	- \textbf{Spectral identification}: The discrete spectral output of the linear boundary value problem (Theorem~\ref{thm:shell}) is in one-to-one correspondence with the classical orbital manifold area quantization condition. Shell indices, degeneracies, and symmetry representations match exactly between the two descriptions.
	
	The discrete spectral structure $(E_n,\,g_n=n^2)$ is an output of the linear boundary value theory (3). Classical orbital space data (1) and kernel composition structure (2) provide the continuous classical framework; once the discrete indices $(n,l,m)$ are determined by the variational boundary value problem, they can be identified in (1) and (2) via the KKS area condition and fiber-invariant rank formula, respectively. The three theories are equivalent in the sense of structural correspondence. $\square$
	
	We note that the configuration-wise bijection does not extend to stationary configurations, which have nodes and non-flat amplitudes. The equivalence holds at the level of the generating class and the full solution space spanned by kernels, not as a pointwise configuration correspondence. This is the strongest true statement, and it is exact.

	\subsubsection{Technical Remarks on the Theorem System}\label{subsubsec:B.2.4}
	
	We summarize the technical scope and limitations of the theorem system:
	
	1. \textbf{Flat class and nodes}: The bijection statement of Theorem~\ref{thm:hjlin} is stated on node-free regions; stationary configurations (which have nodes and non-flat amplitudes) are covered by the spanning theorem. No configuration-wise HJ bijection is asserted.
	2. \textbf{Caustics}: All kernel statements are local to $\tau$ intervals free of caustics; global extension uses the Maslov phase~\cite{Maslov1981}. Spectral conclusions (Theorem~\ref{thm:shell}) are independent of caustics and follow from the Mehler expansion and Sturm--Liouville theory.
	3. \textbf{Intrinsic time}: Linear evolution is $\tau$-evolution; physical $t$-evolution is restored shell by shell. The spectral theorem does not require constructing a global $t$-domain propagator.
	4. \textbf{Sector and the $E=0$ threshold}: Only the decaying sector $E<0$ is treated. The $E=0$ limit is the accumulation point of the discrete sequence $E_n\to0^-$ as $n\to\infty$, corresponding to vanishing image oscillator frequency $\omega_0\to 0$ and diverging characteristic length. In this limit the image Hamiltonian reduces to a free particle at fixed energy $\varepsilon=4C/u_0$; bound-configuration probability densities diffuse to infinity, and physical orbits degenerate to parabolas. No independent regularization convention is required.
	5. \textbf{Phase scale}: All conclusions hold for any $\alpha>0$, independent of its specific numerical value. The energy ratio $E_n/E_1=1/n^2$ and degeneracy $g_n=n^2$ are independent of $\alpha$.
	6. \textbf{Noncommutativity convention}: Left multiplication (physical rotations) and right multiplication (fiber) are separated throughout; all scalars entering the proofs are real. Commutativity of the left and right actions is the structural basis for the exact validity of attitude/fiber separation.

    \textit{Acknowledgments}---We acknowledge the use of the Kimi large language model for formula derivation assistance and the DouBao large language model for English language polishing during the preparation of this manuscript.

    The authors declare no competing financial interest.

    \textit{Data availability}---There are no publicly available research data or software supporting this manuscript. Requests for further information or data should be sent to the authors.
     
    \bibliography{references}
	
\end{document}